\documentclass[sigconf]{acmart}

\usepackage[english]{babel}
\usepackage{blindtext}
\usepackage{svg}
\usepackage{tikz}
\usepackage{amsmath}
\usepackage{xcolor}

\usepackage[toc,page]{appendix}

\usepackage[linesnumbered,ruled,vlined]{algorithm2e}
\usepackage{amsthm}

\newcommand{\soheiltitle}[1]{\noindent\textbf{#1:}}

\newenvironment{soheillist}
  {
  \begin{list}
    {
    }
    {
     \setlength{\labelwidth}{-0.25em}
     \setlength{\labelsep}{.75em}
     \setlength{\itemsep}{2pt}
     \setlength{\leftmargin}{0.5cm}

    }
  }
{\end{list}}
\newtheorem*{theorem*}{Theorem}

\usepackage{filecontents}
\usepackage{lscape}
\usepackage{tikz}
\usepackage{pdfpages}
\usepackage{subcaption}
\usepackage{graphicx}
\usepackage{amsmath}

\DeclareMathOperator{\EX}{\mathbb{E}}

\usepackage{tikz}
\usetikzlibrary{plotmarks}
\usepackage{pgfplots}
\usepackage{blindtext}
\renewcommand\footnotetextcopyrightpermission[1]{} 
\setcopyright{none}

\acmDOI{}

\acmISBN{}

\acmPrice{}

\begin{document}


\title[A Non-Deterministic CC Approach for 5G and Beyond]{A Simple Non-Deterministic Approach Can Adapt to Complex Unpredictable 5G Cellular Networks}



\author{Parsa Pazhooheshy}
\affiliation{%
  \institution{University of Toronto}
}
\email{parsap@cs.torondo.edu}

\author{Soheil Abbasloo}
\affiliation{%
  \institution{University of Toronto}
}
\email{abbasloo@cs.toronto.edu}

\author{Yashar Ganjali}
\affiliation{%
  \institution{University of Toronto}
}
\email{yganjali@cs.toronto.edu}


\begin{abstract}
5G cellular networks are envisioned to support a wide range of emerging delay-oriented services with different delay requirements (e.g., 20ms for VR/AR, 40ms for cloud gaming, and 100ms for immersive video streaming).
However, due to the highly variable and unpredictable nature of 5G access links, existing end-to-end (e2e) congestion control (CC) schemes perform poorly for them. In this paper, we demonstrate that properly blending \textit{non-deterministic}
exploration techniques with straightforward \textit{proactive} and \textit{reactive} measures is sufficient to design a simple yet effective e2e CC scheme for 5G networks that can: (1) achieve high controllable performance, and (2) possess provable properties. To that end, we designed Reminis and through extensive experiments on emulated and real-world 5G networks, show the performance benefits of it compared with different CC schemes. 
For instance, averaged over 60 different 5G cellular links on the Standalone (SA) scenarios, compared with a recent design by Google (BBR2), Reminis can achieve $2.2\times$ lower 95th percentile delay while having the same link utilization.


\end{abstract}

\maketitle

\section{Introduction}

Congestion Control (CC), as one of the active research topics in the network community, has played a vital role during the last four decades in satisfying the quality of service (QoS) requirements of different applications~\cite{newreno,cubic,bbr2,vegas,orca,sage,natcp,c2tcp,sprout,copa,deepcc,verus,vivace,c2tcp_conf}. Although most of the early efforts for designing CC schemes targeted general networks with their general characteristics, as time went by and new network environments emerged, the idea of designing environment-aware CC schemes showed its advantages (e.g., TCP Hybla \cite{hybla} and TCP-Peach \cite{peach} for satellite communication with its unique loss properties, PCCP \cite{PCCP} and 
TARA \cite{tara} for sensor networks with their unique resource restrictions, and DCTCP~\cite{dctcp} and 
TIMELY \cite{timely}
for data center networks (DCN) with their unique single-authority nature).

One of the important emerging network environments with huge potential is the 5G cellular network. Just in the first quarter of 2022, the number of connections over 5G reached more than 700 million, while it is expected that by the end of 2026, this number will surpass 4.8 billion globally~\cite{5G_americas}. Considering such a huge increase in the number of 5G users, the wide variety of current and future applications, and the range of new network characteristics and challenges it brings to the table, the need for a 5G-tailored CC scheme reveals itself.

\subsection{What Makes 5G Different?}
\label{sec:intro:5g_is_diff}

\par
\noindent\textbf{Orders of Magnitude Larger Bandwidth-Delay Product:}
Recent measurements have shown that current millimeter-wave (high-band) 5G networks can achieve, on average, $\approx$1 Gbps link capacities (and up to 2 Gbps) and around 20 ms e2e delays~\cite{xu2020understanding,narayanan2020lumos5g}. Compared to a DCN with 100Gbps access links and $10\mu s$ e2e delay, 5G networks can have, on average 40$\times$ larger bandwidth-delay product (BDP). Compared to its predecessor, e.g., a 4G network with a 20 Mbps link and 40 ms e2e delay, 5G networks have on average 50$\times$ larger BDP.


\noindent\textbf{Highly Variable \& Unpredictable Access Links:} 
One distinguishing characteristic of 5G links compared to its predecessors, is the wide range of link capacity fluctuations. While 5G links can reach a capacity as high as 2 Gbps, they can quickly drop below 4G link capacities or even to nearly zero (5G “dead zones”)~\cite{narayanan2020lumos5g}. For instance, the standard deviation in 120 5G link capacities collected in prior work is around 432 Mbps~\cite{narayanan2020lumos5g}~\footnote{For example, considering 4G/3G traces gathered in prior work \cite{deepcc,c2tcp}, this value is about two orders of magnitude larger than its 4G/3G networks' counterpart}. 
\color{black}
The changing dynamics of the environment such as user mobility, environmental obstacles, and other 5G network factors like network coverage, 5G cell size, and handover are among some of the main reasons for these highly unpredictable fluctuations in the link capacity.
\noindent\textbf{Emerging Applications with Unique Delay Requirements:} 
5G networks are envisioned to serve a diverse set of emerging delay-sensitive applications such as AR/VR, online gaming, vehicle-to-vehicle communications, tactile Internet, remote medical operations, and machine learning-enabled services. 
Cellular providers have already started deploying some of these applications such as cloud gaming, real-time augmented/virtual reality (AR/VR), and immersive video streaming at the edge of their networks~\cite{5gapps}. 
Each of these applications has different delay constraints. For instance, the acceptable delay for AR/VR is about 20ms \cite{ar_vr}, for cloud gaming, the target delay is about 40-60ms, and for immersive video streaming the delay shall not be more than 100ms~\cite{5g_use_cases}. 


\noindent\textbf{New Features, New Opportunities:}
The goal of supporting delay-sensitive applications in 5G has popularized the integration of 5G and edge-based frameworks such as mobile edge computing (MEC) ~\cite{mec}. Such an edge-based integration, when combined with other new technologies such as 5G network slicing~\cite{5g_slicing} can provide interesting new design opportunities for CC design. For instance, in an edge-based cellular architecture where users have their own isolated logical networks on top of the shared infrastructures, the concern of TCP-friendliness becomes less of an issue for a 5G-tailored CC scheme.\\


\subsection{Impact of 5G's Unique Properties on CC}
\label{others_fail}

\begin{figure}[!pt]
    \centering
    \includegraphics[width=0.9\linewidth]{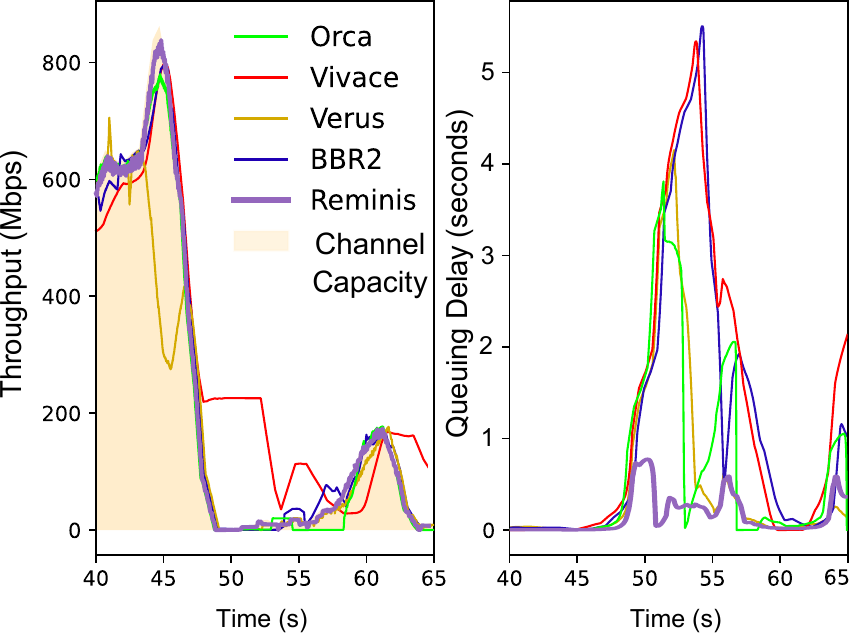}
    \caption{Performance of state-of-the-art CCA on a slice of a sample 5G trace gathered by prior work~\cite{narayanan2020lumos5g}}
    \label{fig:motivation} 
\end{figure}

As a simple motivating experiment, here (following the setting described in Section~\ref{sec:eval:setting}) we use a 5G cellular trace (gathered in prior work\cite{narayanan2020lumos5g}) and report the delay and throughput performance of some recent CC schemes namely BBR2~\cite{bbr2} (representing white-box approach), Verus~\cite{verus} (representing 4G-tailored schemes), Orca~\cite{orca} (representing RL-based designs), and Vivace~\cite{vivace} (representing online-learning designs) over a 20s period of this trace. 
As Fig. ~\ref{fig:motivation} illustrates, the 5G link capacity falls from nearly 1 Gbps to about zero in just 3 seconds (a sample of high variability and unpredictability of 5G links). In this setting, BBR2 faces a clear delay issue to the extent that it generates more than 5 seconds of queuing delay. A white-box CC approach such as BBR2 assumes that the network always follows a certain model. However, when it faces an unpredictable 5G link that clearly diverges from BBR2's wired model of the network, BBR2 cannot adapt to the dynamics of the network quickly and fails to deliver the desired performance. On the other hand, Verus, a design targeting 4G cellular networks, is a black-box approach and tries not to rely on a pre-built model of the network. However, it is very slow, fails to keep up with the available link capacity, and ends up with low link utilization and high queuing delay. 
Considering the performance of Orca and Vivace, it is clear that the learning-based schemes have a hard time in this setting as well. Although Orca, as one of the state-of-the-art reinforcement learning (RL) based schemes, performs better than BBR2 and Verus, it still can experience large queuing delays. We think that the reason lies in the generalization issue of the RL-based designs and the fact that Orca has not seen these scenarios during its training phase. However, how to train an RL-based scheme to achieve high generalization is still a big unanswered question~\cite{orca}. Vivace addresses the issue of the need for offline training by exploiting online-learning techniques.
However, when utilized in a 5G setting, Vivace cannot keep up with the unpredictable and fast-varying nature of the network. For instance, after the drop of capacity at around 47s (in Fig.~\ref{fig:motivation}), it takes more than 10 seconds for Vivace to adjust its sending rate to a value lower than the link's capacity. \color{black} 

\subsection{Design Decisions}
Putting all together, existing general-purpose heuristics (e.g., BBR2), 4G specialized heuristics (e.g., Verus), and convoluted learning-based designs (e.g., Orca) cannot achieve high performance in 5G networks, especially when emerging delay-sensitive applications are considered.\footnote{Appendix~\ref{sec:app:related} briefly overviews some of the related works.} This sheds light on why in this work, we are motivated to design a performant CC scheme for one of the fastest-growing means of access to the Internet, 5G cellular networks. 
To that end, we target two main properties for our design:
\begin{soheillist}
    \item Simplicity \& Interpretability: In contrast with tangled learning-based schemes, we seek a simple design that is easy to reason about and possesses provable properties. 
    \item Adaptability: Showing the performance shortcomings of the existing CC heuristics in 5G networks, our goal is to design a CC scheme that can effectively and efficiently adapt to the dynamics of complex 5G networks and achieve high performance in terms of throughput and delay. 
\end{soheillist}
Favoring simplicity and interpretability led us to avoid employing convoluted techniques such as learning-based ones in this work. Instead, we go back to simple and intuitive principles to design an effective heuristic (called Reminis) that can adapt to highly variable 5G networks\footnote{This work does not raise any ethical issues.}. 
In particular, examples such as the one shown in Fig.~\ref{fig:motivation} indicate that in highly variable and unpredictable 5G cellular links, gaining high utilization requires agile mechanisms to cope with the sudden increase in the available link capacities while achieving low controlled e2e delay requires effective fast proactive and reactive techniques to avoid bloating the network when link capacities suddenly decrease. Based on these observations, Reminis utilizes two key techniques: (1) \textit{non-deterministic explorations} for discovering suddenly available link capacities, and (2) \textit{fast proactive and agile reactive slowdowns} to avoid bloating the network.
As Fig.~\ref{fig:motivation} illustrates, using these intuitive strategies enables Reminis to effectively achieve high throughput while keeping the e2e delay of packets very low. Sections~\ref{sec:design_overview} and~\ref{detailed_design} elaborate more on these design decisions.


\color{black}

\subsection{Contributions}
\label{sec:contribution}
Our key contributions in this paper are as follows.
\begin{soheillist}
\item By designing Reminis, we demonstrate that without the need to use convoluted learning techniques or prediction algorithms, using lightweight yet effective techniques can lead to promising performance on 5G links.
    
\item We mathematically analyze Reminis and prove that it converges to a steady state with a bounded self-inflicted queuing delay that can be controlled in an e2e fashion.
    
\item Through extensive experiments over various emulated 5G cellular traces and a real-world deployed 5G network in North America, we illustrate that Reminis can adapt to highly unpredictable 5G cellular links~\footnote{For example, in our emulations over 60 different 5G traces, when compared to a recent work by Google, BBR2, Reminis can achieve up to 3.6$\times$ and on average $2.2\times$ lower 95th percentile delay while having the same link utilization as BBR2.}.
    
\item As a side effect of our efforts to evaluate Reminis and make reproducible experiments, we debugged and improved  Mahimahi~\cite{mahi}, which cannot emulate high-capacity 5G links\footnote{For more details, see Appendix~\ref{appendix:MM_patch} and discussions therein.}. Our Mahimahi patch along with Reminis' framework are publicly available to facilitate the community's further research on 5G CC \cite{remi_git_repo}. 
\end{soheillist}

\begin{figure}
\begin{center}
\includegraphics[width=0.9\linewidth]{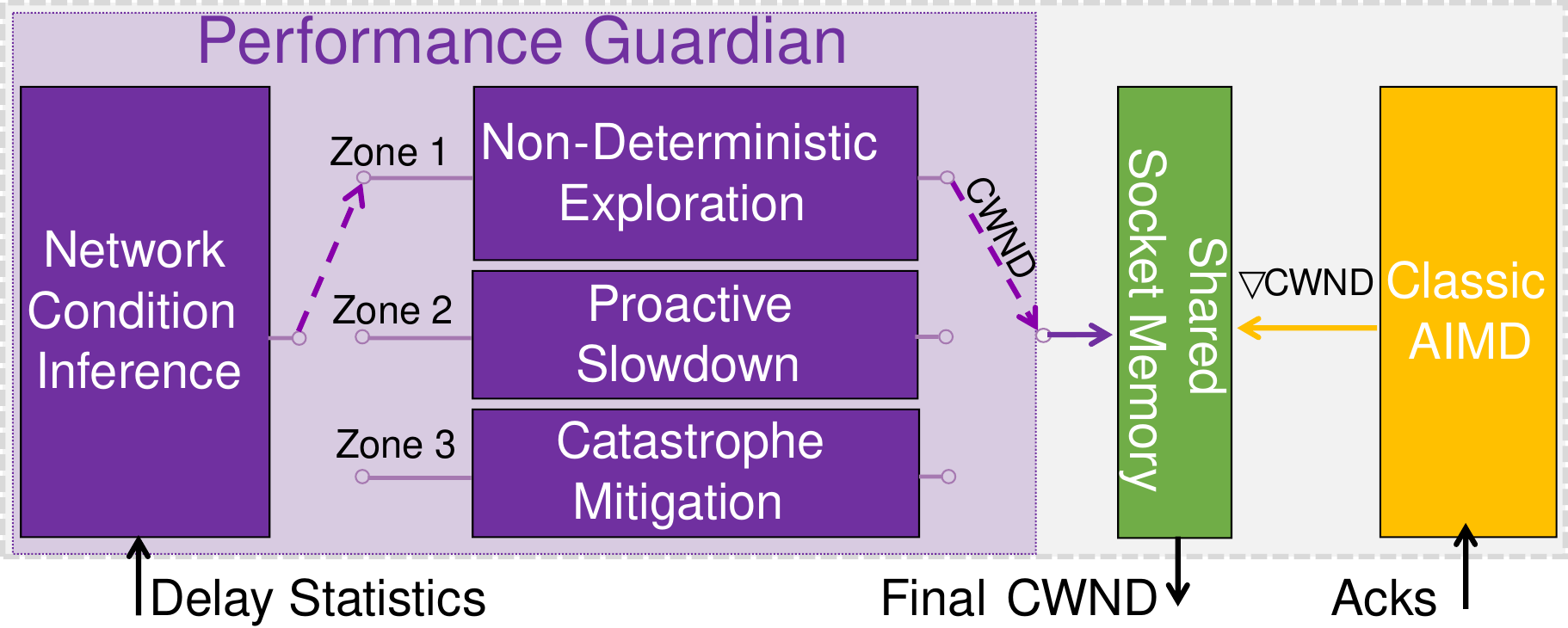}
\end{center}
\caption{ Reminis High-Level Block Diagram}
\label{bb}
\end{figure}

\section{Design Overview}
\label{sec:design_overview}
As shown in Fig.~\ref{bb}, Reminis is composed of two main components: (1) the Classic AIMD unit, an Ack-triggered logic performing the well-known AIMD behavior upon receiving Ack packets, and (2) the Performance Guardian module (or Guardian in short) which runs periodically and adjusts CWND to keep delay low while maintaining high throughput. 

Upon activation in each period, the Guardian exploits a two-step logic. In the first step, the Network Condition Inference (NCI) module, utilizes the history of delay statistics (i.e., delay  and its first-order derivative) 
and infers the current condition of the network. NCI only uses simple e2e RTT samples as input. As we show later in sections~\ref{sec:design:nci} and~\ref{sec:eval}, the simple e2e RTT sample is sufficient for distinguishing different network conditions. 
In the second step, based on the inferred network condition, the Guardian activates one of the following three modules: (1) Non-Deterministic Exploration (NDE), (2) Proactive Slowdown (PS), or (3) Catastrophe Mitigation (CM). 

In particular, if inferred network condition suggests that there is a potential room for gaining higher throughput, the Guardian activates the NDE module to discover further unutilized network bandwidth. On the other hand, if NCI  expects 
an unwanted delay surge in the near future, the Guardian activates the PS module to proactively reduce the chance of a future increase in the delay. If the proactive measures had not been successful and the observed delay has already increased significantly, as the last measure, the CM module is activated. The CM block has a reactive logic and enforces a dramatic decrease of CWND to avoid further increases in delay. In Section~\ref{detailed_design} we elaborate more on the details of these modules, their effectiveness, and their necessity in adapting to highly variable 5G access links and steering Reminis to very high performance. 

\textbf{Why Guarding Periodically and not Per-Packet?}
There are two important reasons behind the periodic nature of the Guardian's task. 
First, as mentioned in \ref{sec:intro:5g_is_diff}, due to the high variability of 5G cellular networks, jumping to any conclusions about network conditions solely based on the statistics of one packet is not reasonable. Hence, it is important to monitor and extract a more stable view of the network by considering more packets. 
In other words, any per-packet measurement in a highly variable network is prone to noise and can lead to inferring wrong network conditions, while having more \textit{samples} from the environment, can potentially help to get a better picture of the network. That is why Reminis utilizes samples observed in periodic intervals. We refer to these intervals as sampling intervals (SI). Second, as discussed earlier, 5G access links can have high capacities. A direct impact of this property on any logic that makes per-packet decisions is a potentially higher CPU utilization compared to periodic logic. 

\textbf{The Role of AIMD Block:}
As discussed, guarding periodically can be helpful in several ways. However, in a highly varying network, only relying on periodic logic can possibly lead to a lack of agility for the system. In other words, a logic that purely relies on periodic samples can react slowly to the changes in link capacity during each SI.
As being agile is significantly vital for CC schemes targeting cellular networks, Reminis harnesses a classic AIMD block to perform extra per-packet reactions during SIs. Later in Section~\ref{sec:eval}, we show this technique not only makes Reminis very agile (e.g., Fig.~\ref{AIMD_fig}) and high performance (e.g., Fig~\ref{fig_general}), but also makes it a very light-weight scheme with very low CPU overhead (Fig.~\ref{cpu_util}) which is another fundamental requirement for a successful CC scheme in high capacity 5G networks.


\color{black}
\section{Reminis Design}
\label{detailed_design}
In this section, we discuss the main components of the Guardian block. First, we introduce the NCI module responsible for inferring network conditions based on delay statistics. Then, we describe the modules responsible for modifying the CWND namely NDE, PS, and CM. 
\color{black}
\subsection{Network Condition Inference (NCI)}
\label{sec:design:nci}
The NCI module uses two signals to infer network conditions: $1)$ delay and $2)$ delay derivative. Delay is measured with end-to-end RTT samples, and as the Guardian runs periodically, once every SI, the delay derivative is defined to be the difference between two consecutive delay values divided by the time difference between two queries.\par
Reminis has a target for its delay which is denoted as the delay tolerance threshold (DTT). This value could be defined based on (or by) the target application or Reminis can use its default value for DTT which is $1.5 \times mRTT$. Here, mRTT is the observed minimum RTT  since starting the flow (which is not necessarily equal to the exact/actual minimum RTT of the network). Considering DTT as the delay target, the NCI module uses the statistics of the delay signals and deduces the network condition in each SI. Later, these inferred conditions will be exploited by NDE, PS, and CM blocks. The length of each SI is equal to the mRTT. 
\color{black}
To keep Reminis simple and lightweight, we use straightforward delay signals for defining network conditions (NCI Zones). In particular, denoting the delay and delay derivative in each $SI_n$ respectively with $d_n$ and $\nabla d_n$, Reminis divides the delay space into three Zones. 
\color{black}
\\ \textbf{\emph{Zone 1}}:  $ d_n \le DTT \: \& \: \nabla d_n  \le 0$, \label{cases-1}
\\ \textbf{\emph{Zone 2}}:  $ d_n \le DTT \: \& \: \nabla d_n >0$, and \label{cases-2}
\\ \textbf{\emph{Zone 3}}:  $ d_n> DTT$. \label{cases-3}

Zone 1 indicates that the delay is below DTT and decreasing, which Reminis interprets as having room for sending more packets for the benefit of getting more throughput. The main reason for this deduction is that the negative delay derivative shows that the sending rate is less than channel capacity and the queue is depleting. Also having a delay of less than DTT gives room for some exploration. \\
Zone 2 shows that the delay is still below DTT but increasing which means that keeping the current CWND or increasing it might lead to a violation in DTT as the positive delay is a sign of queue building up. \\
Finally, Zone 3 indicates that the sender CWND should decrease harshly as being in this Zone means that the delay has exceeded DTT. Many reasons such as 5G dead zones, which are common in 5G networks, can result in transitioning into this Zone.
\begin{equation}
\footnotesize
\label{Safe_Zone euqation}
\texttt{\texttt{{SafeZone}}} (d)= 1 - \dfrac{d - mRTT}{DTT- mRTT}
\end{equation}
\par
In order to quantify how much delay has exceeded DTT, the Guardian uses a function called \texttt{SafeZone}, defined as in Equation \ref{Safe_Zone euqation}.
Based on the zone inferred by the NCI module, one of the NDE, PS, or CM modules will be activated. 
Algorithm \ref{alg:1}, shows one iteration of the NCI module's logic.
\begin{algorithm}[!h]
\caption{The Guardian}\label{alg:1}
\small
 d\_derivative = (d\_now - d\_prev)/{interval}; \\
 cum\_d\_der += d\_derivative;\\
 sz\_now = \texttt{SafeZone}  (d\_now );\\
 
 \uIf{sz\_now < 0}{
     \texttt{{CatastropheMitigation}} (sz\_now);\\
  }
  \uElseIf{d\_derivative > 0}{
    \texttt{ProactiveSlowdown} (d\_now, d\_derivative);\\
  }
\uElseIf {d\_derivative < 0 }{
   \texttt{{NDExploration}} (cum\_d\_der)\;
  }
d\_prev = d\_now;\\
\end{algorithm}
\begin{figure*}
    \centering
    \begin{minipage}[b]{0.33\linewidth}
        \centering
       \includegraphics[width=0.9\linewidth]{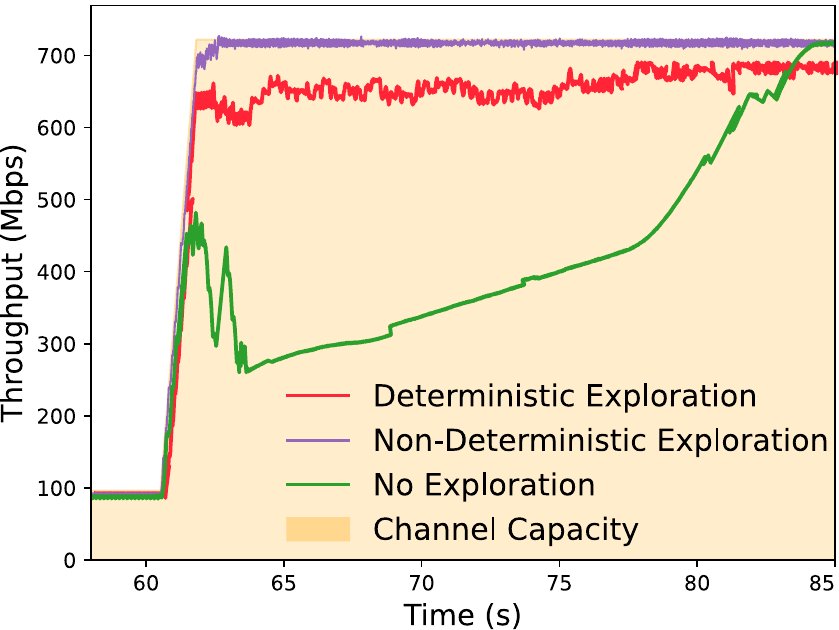}
       \caption{NDE Block in Action}
       \label{fig:design:nde}    
    \end{minipage}
        \begin{minipage}[b]{0.33\linewidth}
            \centering
       \includegraphics[width=0.93\linewidth]{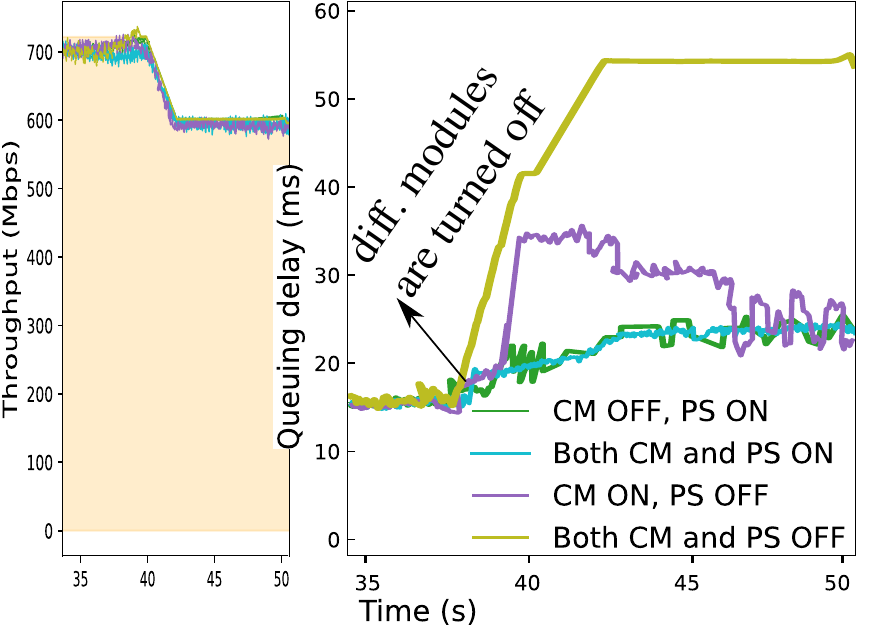}
       \caption{PS Block in Action}
       \label{fig:design:ps}    
    \end{minipage}
        \begin{minipage}[b]{0.33\linewidth}
            \centering
        \includegraphics[width=0.93\linewidth]{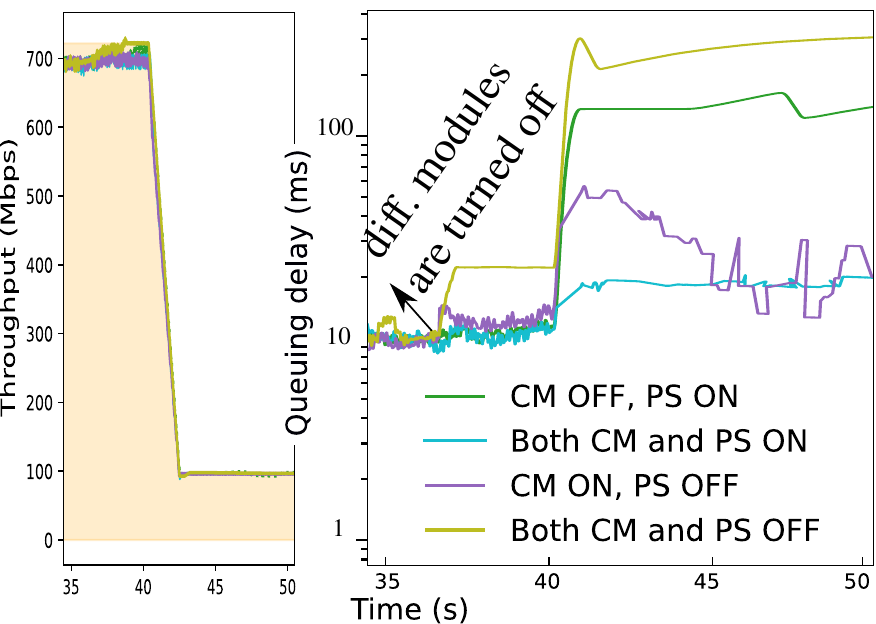}
       \caption{CM Block in Action}
       \label{fig:design:cm}    
    \end{minipage}
\end{figure*}
\subsection{Non-Deterministic Exploration (NDE)}
\label{desing::NDE}
On the one hand, as discussed earlier, being in Zone 1, indicates there is room for sending more packets for the benefit of getting more throughput. On the other hand, many different factors such as user mobility, dynamics of the physical obstacles, the wireless scheduler's algorithms that divide resources between users through time, etc. make 5G access links highly unpredictable. This means although Zone 1 can indicate a possible chance to gain more throughput, it cannot identify the exact amount of such an increase. 

In that landscape, the NDE module is responsible for discovering and utilizing available capacity, without the risk of bufferbloat. To that end, when Zone 1 is inferred by the NCI module, the Guardian activates the NDE module to explore different CWND values in a non-deterministic fashion so that it can address the unpredictable nature of the available link capacities. This can help Reminis utilize the sudden unpredictable surges in access link capacity in a better way. 

However, it's important to make sure that when NDE is exploring any available link capacity, it does not bloat the user's queue in an uncontrollable manner. To address this issue, NDE block controls the average of the stochastic decision-making process with regard to the general trend of the link capacity. In particular, the exploration needs to be more aggressive if the Guardian has been measuring high negative delay derivatives as more negative delay derivative indicates that the sending rate is far less than channel capacity.\color{black}
\par
To this end, the NDE module maintains a Gaussian distribution, $\mathcal{N}(\mu_{n},\,\sigma^{2}_{n})$, where the mean and variance of this distribution change by each delay derivative measured in every SI based on $\mu_{n} \leftarrow \mu_{n-1} - \nabla d_n$ and  $\sigma^{2}_{n} = \frac{\mu_{n}}{4}$ update rules.
Upon activation, the NDE module draws one sample from the Gaussian distribution, $x \sim \mathcal{N}(\mu_{n},\,\sigma^{2}_{n})$, and feeds the sample to a Sigmoid function, $S(x) = \frac{1}{1+ e^{-x}}$. Then, the output of the Sigmoid function is used to increase the current CWND by multiplying the current CWND by $2^{S(x)}$ as shown in Equation~\ref{ndem}.
\begin{equation}
\small
\label{ndem}
    cwnd_n \leftarrow cwnd_n \times 2^{S(x)}
\end{equation}

\par
The range of the Sigmoid function is $(0,1)$, so the NDE module will increase the CWND by a factor between 1 and 2. If Reminis starts measuring negative delay derivatives consecutively, the incremental factor generated by this module will be close to 2, which helps Reminis to adapt to any increase in link capacity quickly. On the other hand, if Reminis measures one negative delay derivative after many positive delay derivatives, the stochastic exploration will be more conservative and increases the CWND by factors slightly more than 1. \par 
We prove that 
the NDE module makes Reminis faster than an AIMD module alone. In particular, assuming that $w_1$ is the CWND that fully utilizes the fixed link without causing any queue build-up, we prove that:
\begin{theorem*}
Reminis helps the AIMD logic to reach $w_1$ in $\mathcal{O}(\log{}w_1)$ instead of $\mathcal{O}(w_1)$ in congestion avoidance phase.
\end{theorem*}
Due to the space limitation, the proof and the detailed assumptions are provided in Appendix~\ref{Speed_up_proof}.

\soheiltitle{NDE in Action}
In a nutshell, the NDE module in Reminis is responsible for tackling scenarios in which the 5G access link experiences sudden surges in link capacity and utilizing these surges in an agile manner. To illustrate the effectiveness of NDE in practice, we use a toy example where the link capacity increases from 100 Mbps to 720 Mbps in a few seconds (Fig.\ref{fig:design:nde}) and compare it with two alternatives: (1) No-Exploration and (2) Deterministic exploration. For the No-Exploration version, we simply turn off the NDE block and for the Deterministic version, we always use the updated mean of the Gaussian distribution instead of using random samples drawn from the Gaussian distribution. 
As Fig.\ref{fig:design:nde} illustrates, without the exploration module, Reminis suffers from heavy under-utilization. 
Deterministic exploration can improve over the No-Exploration version; however, the sending rate still converges to the channel capacity very slowly. In contrast, the NDE block enables Reminis to converge to the new channel capacity very fast. 


\color{black}
\subsection{Proactive Slowdown (PS) and \\Catastrophe Mitigation (CM)}
Considering the high fluctuations of the 5G access links, the main role of PS and CM modules is to effectively control the e2e delay without causing significant underutilization.

\color{black}
\soheiltitle{Proactive Slowdown} 
\label{desing::PS}
This module is activated whenever the NCI infers Zone 2. When being in Zone 2, Reminis needs to be prudent so it can prevent any violation of DTT in the next SI. The PS module decreases the current CWND if the delay gets too close to DTT. To detect when the delay is close to DTT (i.e., risk of DTT violation), PS calculates the expected delay in the next SI, using a first-order regression predictor, as in Equation~\ref{prediction}.

\begin{equation}
\small
\label{prediction}
d_{n+1} = d_{n} + \nabla d_{n} \times SI
\end{equation}
\par
Equation \ref{FCP} shows how the PS module decreases the CWND upon activation. The main responsibility of this module is to reduce the CWND if the calculated expected value of delay in the next SI is more than DTT. This module will be harsher in decreasing the CWND if its expectation for DTT violating in the next SI gets larger.
\begin{equation}
\small
    \label{FCP}
    cwnd_n \leftarrow cwnd_n \times 2^{min (0,{\texttt{SafeZone}}( d_{n+1}))}
\end{equation}

\soheiltitle{Catastrophe Mitigation}
Many reasons, such as sudden decreases in link bandwidth, can cause Reminis to end up in Zone 3 despite the PS module actions. In these types of scenarios, we want Reminis to decrease the delay as soon as possible to meet the delay requirement. Therefore, upon the inference of Zone 3 by the NCI, the CM module will be activated.
CM decreases the CWND by at least half upon activation. The decrease would be harsher in proportion to DTT violations.
Equation \ref{dca} shows the CWND update rule by this module.
\begin{equation}
\label{dca}
    cwnd_n \leftarrow cwnd_n \times 2^ {\texttt{SafeZone} ({d}_n)}\times 0.5
\end{equation}
Algorithm \ref{alg:2}, describes the NDE, PS, and CM modules' logic. 

 
\begin{algorithm}[!h]
\caption{NDE, PS, and CM}\label{alg:2}
\small
\SetKwProg{Fn}{Function}{:}{}
\Fn{NDExploration(cum\_d\_der)}{
  $\mu = \text{cum\_d\_der}$\;
  $\sigma^{2} = \frac{\mu}{4}$\;
  $x \sim \mathcal{N}(\mu,\,\sigma^{2})$\;
  $cwnd \leftarrow cwnd \times 2^{\frac{1}{1+e^{-x}}}$\;
  \KwRet\;
}
\Fn{ProactiveSlowdown(d\_now, d\_derivative)}{
  $expected\_d\_nxt =  d\_now + (d\_derivative \times \text{interval})$\;
  $expected\_sz\_nxt = \text{Safe\_Zone}(expected\_d\_nxt)$\;
  \uIf{$expected\_sz\_nxt < 0$}{
    $cwnd \leftarrow cwnd \times 2^{expected\_sz\_nxt}$\;
  }
  \KwRet\;
}
\Fn{CatastropheMitigation(sz\_now)}{
  $cwnd \leftarrow cwnd \times 2^{sz\_now} \times 0.5$\;
  \KwRet\;
}
\end{algorithm}

\soheiltitle{PS and CM in Action}
5G access links experience sudden drops due to several reasons described in \ref{sec:intro:5g_is_diff}. These drops could cause a significant increase in delay and violate the delay requirements of 5G applications as a result. So, here, we use two examples to show how PS and CM blocks help Reminis effectively control delay over varying 5G access links. In both examples, we turn off different blocks right before the changes in the capacity and capture their impact on the performance of Reminis. In particular, in the first example, we focus on scenarios where the decrease in access link capacity is relatively small. Fig.~\ref{fig:design:ps} shows such a scenario where the link capacity decreases from 720 Mbps to 600 Mbps. As Fig.~\ref{fig:design:ps} depicts,
in this scenario, the PS block is sufficient to control the queuing delay during the transition. Here, the PS block reacts to the decrease of link capacity by reducing CWND according to Equation \ref{FCP}. This enables Reminis to keep the delay below the DTT and decrease the 95th percentile of queuing delay by 2.3$\times$ compared to the case where the PS module is turned off. 

In the second example, we focus on sudden relatively large decreases in link capacity. For example, Fig.~\ref{fig:design:cm} 
shows a scenario where the link capacity decreases from 720 Mbps to 100 Mbps. As it is clear from Fig.~\ref{fig:design:cm}, the PS module alone is not sufficient and the CM block becomes a key component to control the delay. In particular, the CM module alone ("CM ON, PS OFF" in Fig.~\ref{fig:design:cm}) can eventually control the delay in this scenario. In fact, when the delay surpasses the DTT, the CM module is activated and decreases the CWND by more than half (based on Equation \ref{dca}). Note that, unlike the PS module which has a proactive nature, the CM module is reactive leading to a temporary surge in delay when the PS module is off. In contrast, when both CM and PS are on, the PS module enhances the performance by controlling the delay during the transition ($[40,45s]$ in Fig.\ref{fig:design:cm}). In sum, the proactive aspect of the PS component combined with the reactive aspect of the CM block boosts the overall performance of Reminis. Even in the steady-state, ($[45,50s]$ in Fig.\ref{fig:design:cm}), this combination benefits Reminis and lowers the delay oscillation.
\subsection{Reminis' Steady State}
One of our main quests was to design a simple CC scheme with provable properties. Considering that, 
we mathematically prove that
the following Theorem summarizes the convergence property of Reminis (considering $q_{th}=DTT$):
\begin{theorem}
On average, Reminis converges to a steady state with a queuing delay no more than $(1+S(\frac{\ln{4}-1}{2BDP})\ln{2})q_{th}$.
\end{theorem} 
The detailed proof of the above Theorem and corresponding assumptions are discussed in Appendix~\ref{analysis}.


\color{black}




\begin{figure*}[t]
    \centering
    \begin{minipage}[b]{0.33\linewidth}
        \centering
         \includegraphics[width=0.97\linewidth]{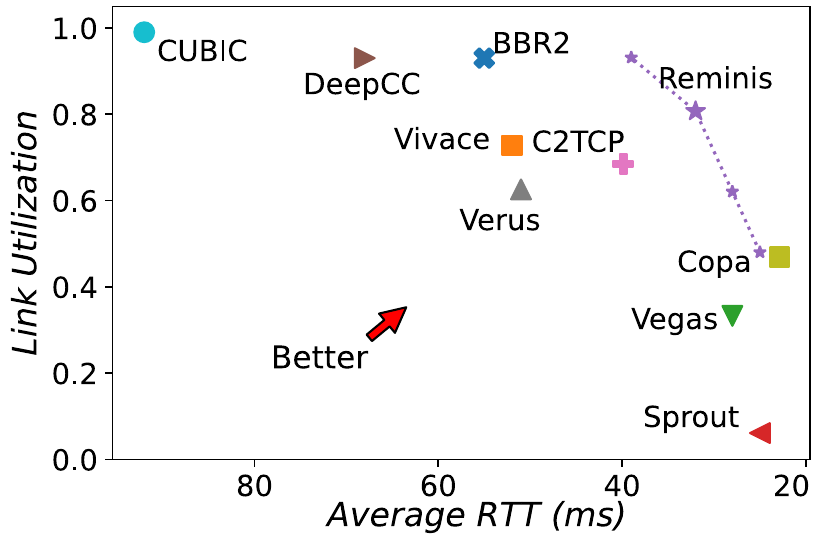}
        \label{fig_aws}
    \end{minipage}        
     \hfill
    \begin{minipage}[b]{0.33\linewidth}
        \centering
        \includegraphics[width=0.97\linewidth]{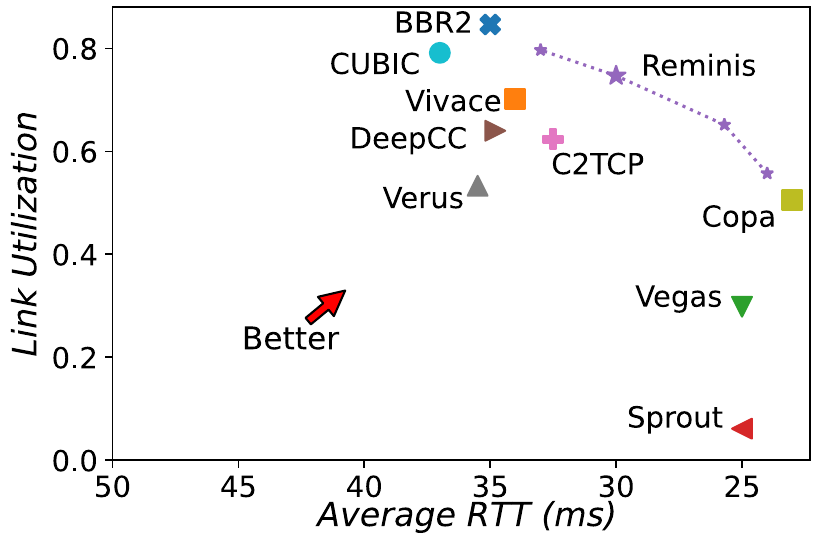}
        \label{fig_geni}
    \end{minipage}
    \hfill
    \begin{minipage}[b]{0.33\linewidth}
        \centering
        \includegraphics[width=0.97\linewidth]{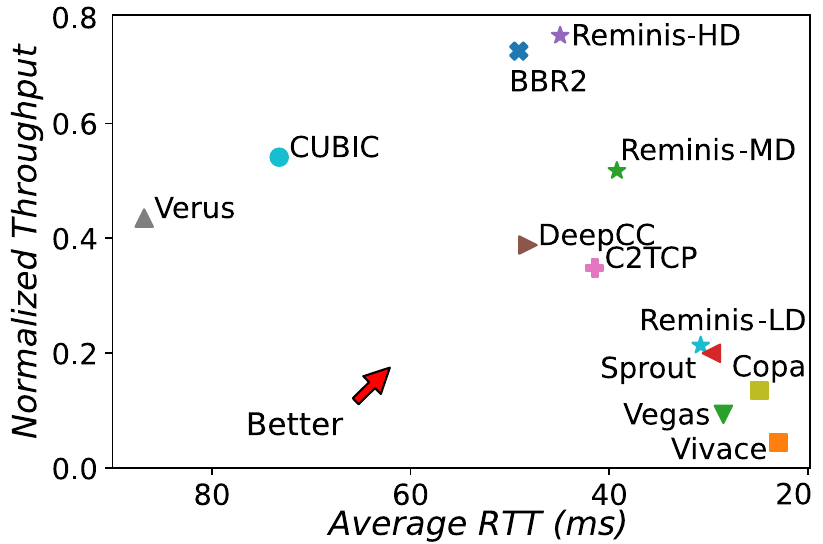}
        \label{fig_cell}
    \end{minipage}
    \hfill
    \begin{minipage}[b]{0.33\linewidth}
        \centering
         \includegraphics[width=0.97\linewidth]{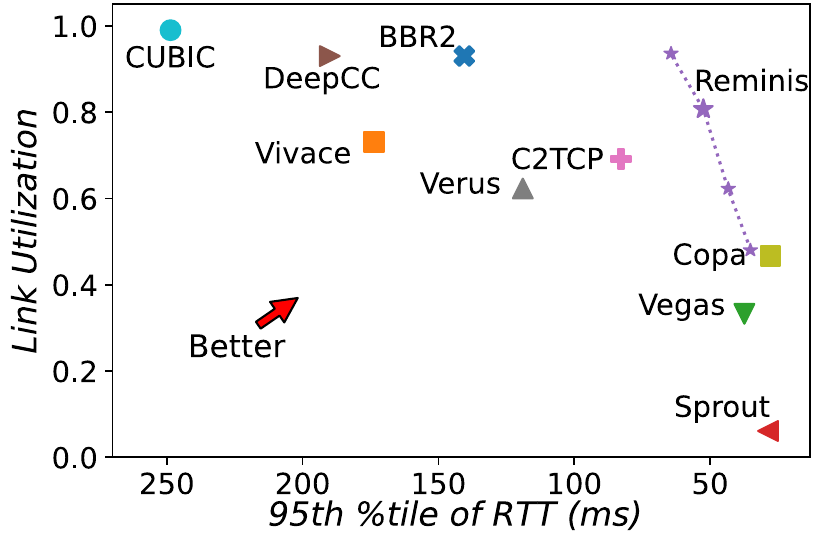}
        \label{fig_aws}
    \end{minipage}        
     \hfill
    \begin{minipage}[b]{0.33\linewidth}
        \centering
        \includegraphics[width=0.97\linewidth]{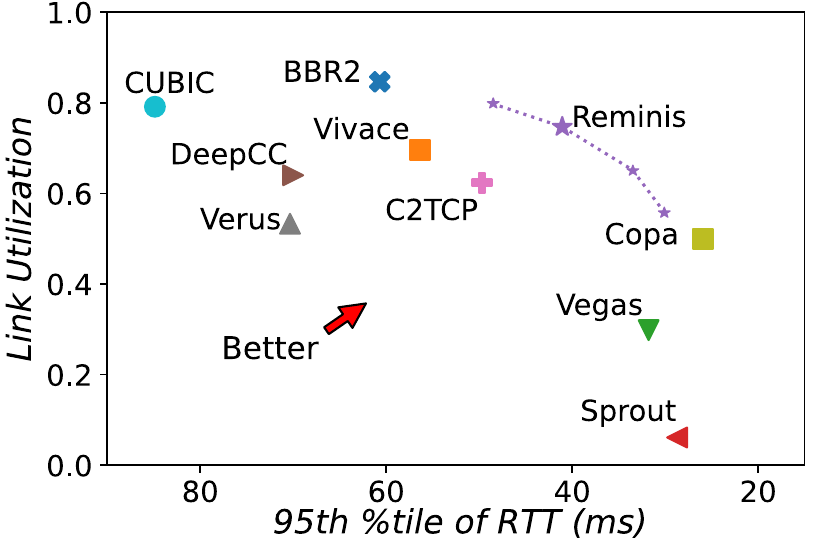}
        \label{fig_geni}
    \end{minipage}
    \hfill
    \begin{minipage}[b]{0.33\linewidth}
        \centering
        \includegraphics[width=0.97\linewidth]{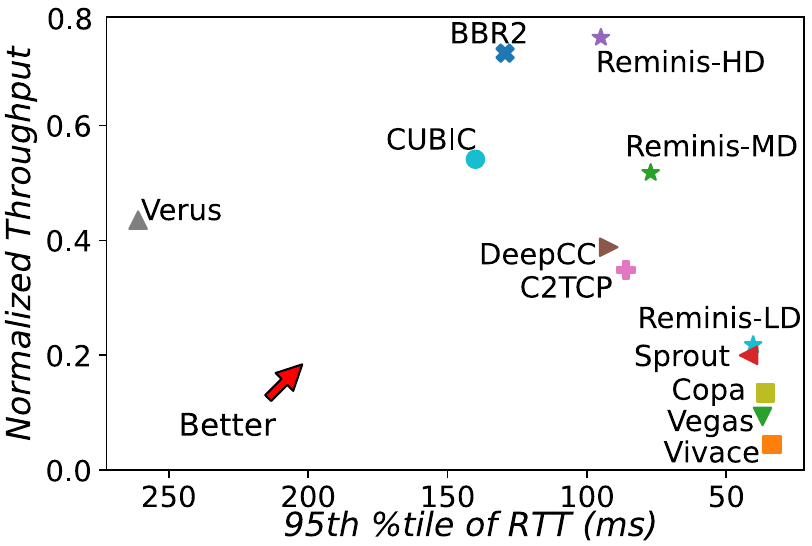}
        \label{fig_cell}
    \end{minipage}
        \caption{Throughput-Delay for SA (left column), NSA (middle column) and In-Field (right column) Emulations.}
      \label{fig_general}
\end{figure*}




\section{General Evaluations}
\label{sec:eval}
In this section, we extensively evaluate Reminis and compare it with other state-of-the-art e2e CC schemes in reproducible trace-based emulations and in-field experiments. The emulations help us to measure Reminis' performance over various scenarios, whereas the infield tests help us verify Reminis' performance in a much more complex real-world network.

\textbf{Metrics}: The main metrics used in this paper, suitable for a real-time application, are the average throughput (or equivalently link utilization) and delay-based statistics such as average and 95th percentile packet delays. \label{metrics}

\textbf{Compared CC Schemes}: We compare Reminis with different classes of state-of-the-art e2e CC schemes. The first class is general-purpose CC schemes such as TCP CUBIC \cite{cubic}, Google's BBR2 \cite{bbr2}, TCP Vegas \cite{vegas}, and Copa \cite{copa}. The second class is CC algorithms that are custom designed for cellular networks. These schemes are C2TCP \cite{c2tcp}, Verus \cite{verus}, and Sprout \cite{sprout}. The final class is learning-based CC schemes. We compare Reminis with DeepCC \cite{deepcc}, targeting cellular networks and PCC-Vivace \cite{vivace} as a general-purpose learning-based CC scheme.

\subsection{Trace-based Emulations}
\label{sec:eval:emulation}
\textbf{Mahimahi Limitation for High Bandwidth Links}: Mahimahi \cite{mahi} was not originally designed to emulate high-speed links. That has led to some design decisions that downgrade its performance in large bandwidth scenarios such as 5G links. We faced these performance issues during the evaluation of Reminis. So, we pinpointed Mahamahi's issues and revised them to support high BDP emulations. In a nutshell, we updated the TUN/TAP settings and logging functionalities of Mahimahi. 
Fig. \ref{appendix:mahi} shows the performance of TCP CUBIC over a 5G link emulated with Mahimahi before and after our changes. As it is clear, after our modifications, Mahimahi is not the performance bottleneck and CUBIC can utilize the link fully. For brevity and sake of space, we omit the details of the changes and refer interested readers to our publicly available source code including these modifications along with the Reminis source code~\cite{remi_git_repo}.

\noindent\textbf{Setup}:
\label{sec:eval:setting}
We use trace-driven emulations to evaluate Reminis and compare it with other CC schemes under reproducible network conditions. We use our patched Mahimahi as the emulator and the 5G traces collected by prior work ~\cite{narayanan2020lumos5g} as our base network traces.
After patching Mahimahi, we evaluate the general performance of Reminis and other relevant CC schemes. For these experiments, we use 60 different 5G traces gathered in North America by prior work \cite{narayanan2020lumos5g} in various scenarios\footnote{For more details about the traces, see Appendix~\ref{appendix:traces} and discussions therein.}. Each run is set to be 3 minutes and we repeat each run 3 times. For these experiments, we fix the 
minimum intrinsic RTT \color{black} of the network to 20ms based on prior measurements done by \cite{xu2020understanding}. 
Furthermore, since currently there are two different deployments of 5G networks, we consider two different settings for bottleneck bandwidth size.
The first deployment is the Non-Standalone (NSA) mode, where operators are reusing the legacy 4G infrastructure to reduce costs. In NSA mode, we expect to have 4G-tuned buffers which would be smaller than the 5G-tuned buffers. For NSA mode, based on measurements done by prior work \cite{xu2020understanding}, we set the buffer size to 800 packets. The second deployment is the Standalone (SA) version, where the infrastructure of the network is also changed to acknowledge 5G networks' needs. In this case, we configure the buffer size to 3200 packets \cite{xu2020understanding}.

\textbf{Standalone (SA) Scenario}: The left column of Fig. \ref{fig_general} shows the performance of tested CC schemes in the SA experiment. The dashed curve shows Reminis' performance with different DTT values. The star with a bigger marker size than other stars is Reminis with its default DTT parameter (i.e. $DTT=1.5 \times mRTT$). Later in Section~\ref{DTT_effect}, we investigate the sensitivity of Reminis to this parameter. Considering either average or 95 percentile delay statistics, pure loss-based CC schemes like CUBIC suffer from high e2e delay, though they can fully utilize the link. This behavior is expected from these schemes as they try to fully occupy the bottleneck buffer. In contrast, Reminis can find a sweet spot in the delay-throughput trade-off. For instance, averaged over all runs and all traces, Reminis with default DTT, achieves 5$\times$ lower 95th percentile delay compared to CUBIC. This promising performance comes with only compromising 20\% of CUBIC's link utilization. Increasing DTT value to $2\times mRTT$ (making Reminis more throughput hungry), Reminis, on average, achieves $2.2\times$ lower 95th percentile delay than BBR2, while having the same link utilization.\\
Moreover, SA-related parts of Fig. \ref{fig_general} show that delay-based CC algorithms like Vegas and Copa cannot get an acceptable link utilization. For example, averaged over all the experiments, default Reminis compared to Vegas gains roughly 2.42$\times$ more throughput while its 95th percentile of delay and average delay are only 1.4$\times$ and 1.14$\times$ more than Vegas respectively. One of the main takeaways from SA scenario experiments is that because of deep buffers, throughput-hungry schemes like CUBIC or BBR2 can fully utilize the link but at the same time, they will have dire delay performance. On the other hand, Reminis, using its Proactive Slowdown and Catastrophe Mitigation modules can control the delay. However, without the Non-Deterministic Exploration module, Reminis would be hindered like other delay-based schemes. Using these modules simultaneously enables Reminis to reach the sweet spot of the delay-throughput trade-off. \par
\textbf{Non-Standalone (NSA) Scenario}: The middle column of Fig. \ref{fig_general}, compares the overall performance of all gathered CC schemes in the NSA scenario. An important note here is that because of the small size of the buffer, even throughput-hungry schemes, cannot fully utilize the link. The highest utilization, in this case, is roughly 80\% achieved by BBR2 and CUBIC. In this scenario, Reminis operated on a sweet spot in the delay-throughput trade-off curve. In particular, having roughly the same link utilization as CUBIC, default Reminis achieves 2$\times$ lower 95th percentile delay than CUBIC. The relative performance of the investigated CC schemes are same as the SA scenario and the only difference is the reduction in delay and link utilization, in general among all CC schemes, as a result of the small buffer size.

\subsection{In-Field Evaluations}
Real-world cellular networks can be more complicated than emulated versions due to the existence of other users, different behavior of cellular base-station packet schedulers, etc. We tested the performance of Reminis over deployed 5G networks
in North America. Having servers as senders, a 5G sim card, and a 5G phone as a client, we collected the performance of various CC schemes under different environments with different dynamics. We used Samsung Galaxy S20 5G as our 5G mobile phone.
The mRTT of the 5G network in our in-field experiments varied from 20ms to 30ms. Overall, we conducted 80 experiments for each CC scheme where each run takes 15 seconds. Experiments are done at different times and places to capture various network dynamics. 
During in-field evaluations, the mobile phone was in both stationary and walking conditions. In both conditions, we observed the time-varying throughput that Reminis targets. In the stationary scenarios, two main reasons cause significant changes in access link capacity over time. The first reason is the changes in line of sight (LoS). Even small obstructions like a human body could trigger 5G-to-4G handoffs and lead to significant performance degradation \cite{LoS}. Second, the 5G wireless scheduler, based on different reasons such as the history of users' resource utilization and the number of current existing users, can enforce different available access link capacities per user. This can cause considerable changes in available link capacity observed by the end user even in the stationary scenario. In a 5G context, demand for resources in different slices can vary arbitrarily and unpredictably. In such cases, a large block of network resources might suddenly be available or taken away from a slice servicing a set of mobile broadband applications \cite{haile2021end}.
\color{black}

For every experiment, the throughput of schemes is normalized to the maximum throughput gained in that specific scenario. In these experiments, we use three versions of Reminis named Reminis high-delay(HD), medium-delay (MD), and low-delay (LD) corresponding to DTT values of 60 ms, 40 ms, and 30 ms, respectively.\\
The right column of Fig.~\ref{fig_general} shows the high-level results for our in-field tests. Reminis-HD achieves the same throughput as BBR2 while achieving on average 1.47$\times$ lower 95th percentile delay. Moreover, Reminis-HD can increase the throughput by 1.34 $\times$ and reduce the 95th percentile RTT  by 1.4 $\times$ compared to TCP CUBIC. 
With a tighter DTT, Reminis-MD can achieve the same throughput as CUBIC while having 1.7$\times$ lower 95th percentile of delay. The results of in-field evaluations are pretty close to NSA scenario emulations, which can verify our assumptions and results for NSA emulations.


\subsection{MEC-Flavored Emulations}


\begin{figure*}[!t]
    \centering
    \begin{minipage}[b]{0.66\linewidth}
        \begin{minipage}[b]{0.49\linewidth}
            \centering            \includegraphics[width=\linewidth,height=1.7in]{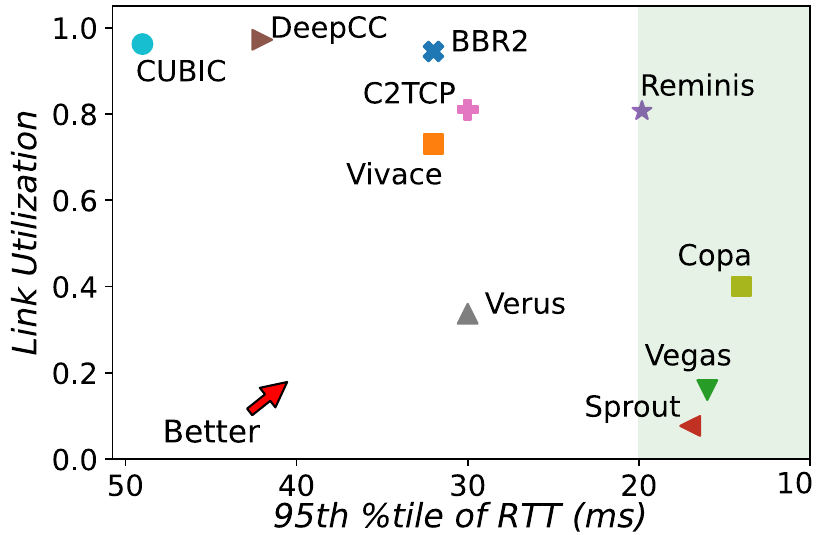}
            \caption{MEC-Flavoured Exp.}
            \label{mec_95}
        \end{minipage}
        \begin{minipage}[b]{0.49\linewidth}
            \centering                  \includegraphics[width=\linewidth, height=1.7in]{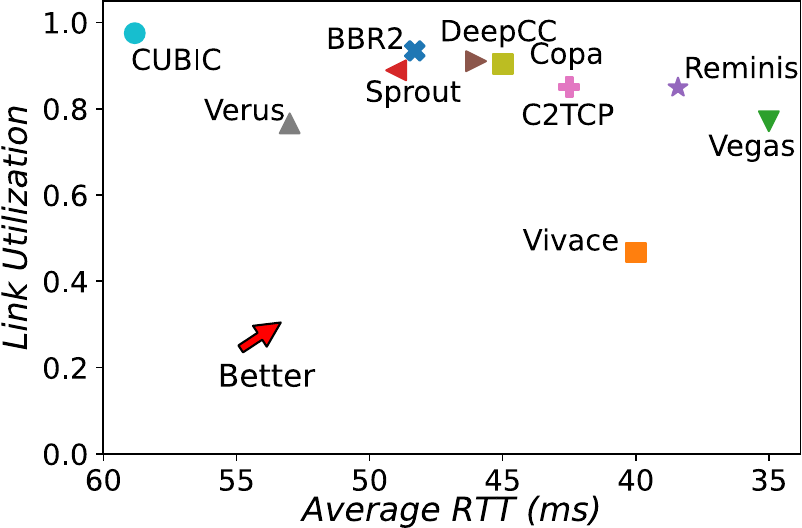}
            \caption{Experiments on 4G Traces}
            \label{fig:eval:4G}
        \end{minipage}
        \end{minipage}
    \begin{minipage}[b]{0.33\linewidth}
       \centering
        \includegraphics[width=\linewidth, height=1.7in]{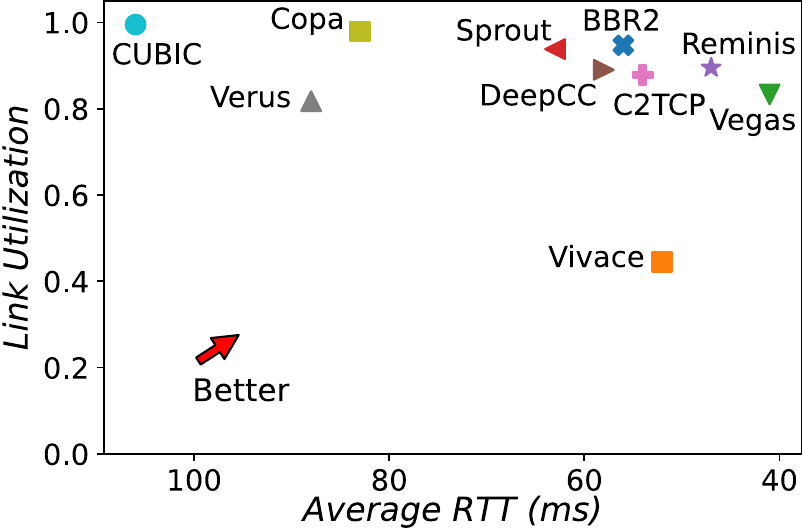}
         \caption{Experiments on 3G Traces}
        \label{fig:eval:3G}
    \end{minipage}
\end{figure*}

\label{sec::MEC}
As mentioned in Section~\ref{sec:intro:5g_is_diff}, mobile edge computing integrated with 5G is one of the design opportunities in 5G networks. 5G aims to support under 10 ms latencies using the New Radio technology. With supporting low-latency connections, applications such as AR/VR are envisioned to be functional in the 5G environment \cite{xu2020understanding}. Here, we emulate a 5G MEC scenario and investigate the behavior of different CC schemes in this scenario. To this end, we set the intrinsic RTT of the network to 10 ms and assume a VR application with a delay constraint of 20 ms is running. Other than this change in intrinsic RTT, we fix the setting to be representative of the SA scenarios. \par
Fig. \ref{mec_95} shows the overall performance of Reminis and other CC schemes in this experiment. As shown, only four schemes can achieve the required latency desired by the AR application (the green area in the figure). However, with the help of the NDE module, Reminis achieves at least 2.4$\times$ more throughput than the other three schemes. This promising performance highlights the benefits of the design decisions of Reminis. Generally, the Non-Deterministic Exploration module helps Reminis to achieve high throughput while the Proactive Slowdown and Catastrophe Mitigation modules help Reminis satisfy its delay target. Appendix Section \ref{app::MEC} gives information about the average delay performance of these CC algorithms in this experiment.

\subsection{Is Reminis Only Good in 5G Networks?}
\label{3G/4G}
    








Here, we show that mechanisms utilized by Reminis are effective to maintain the e2e delay low and adapt to the channel capacity variations not only in 5G cellular networks but also in other networks such as 3G and 4G cellular networks. 

To that end, we use various 3G and 4G traces (gathered respectively by \cite{c2tcp} and \cite{deepcc}) and evaluate the performance of different schemes. A few samples of these traces are shown in Section \ref{fig::trace_sample}. Fig.~\ref{fig:eval:4G} and ~\ref{fig:eval:3G} show the results of these evaluations.\footnote{More results on the improvements of 95th percentile of delay in these emulations are explained in appendix Section \ref{app::3G/4G}.}
There are two important remarks here. 
First, Reminis still performs very well on both 3G and 4G scenarios. For instance, compared to BBR2, on 4G and 3G traces, Reminis achieves 1.48$\times$ and 1.33 $\times$ lower average queuing delay respectively, while BBR2's throughput is only 1.1 $\times$ and 1.05 $\times$ more than Reminis.
Second, the performance gap between other CC schemes and Reminis in 3G and 4G scenarios is smaller compared to the 5G scenarios. The main reason for that is the fact that 5G networks have an order of magnitude larger BDP, deeper buffers, and more volatile access links compared to 3G and 4G networks. This means in the 5G setting, wrong actions of a CC scheme have higher chances of being penalized more and manifest in performance issues. 
\color{black}
\section{Deep Dive Evaluations}
\label{sec::mircro}
In this section, we will look under the hood and investigate the dynamics of Reminis and the role of its individual components. We will also investigate the impact of different parameters such as intrinsic RTT, buffer size, and DTT on Reminis. Finally, we will end this section by examining Reminis' fairness and overhead aspects.

\begin{figure*}
    \begin{minipage}[b]{0.24\linewidth}          
        \centering
        \includegraphics[width=0.95\linewidth, height=1.5in]{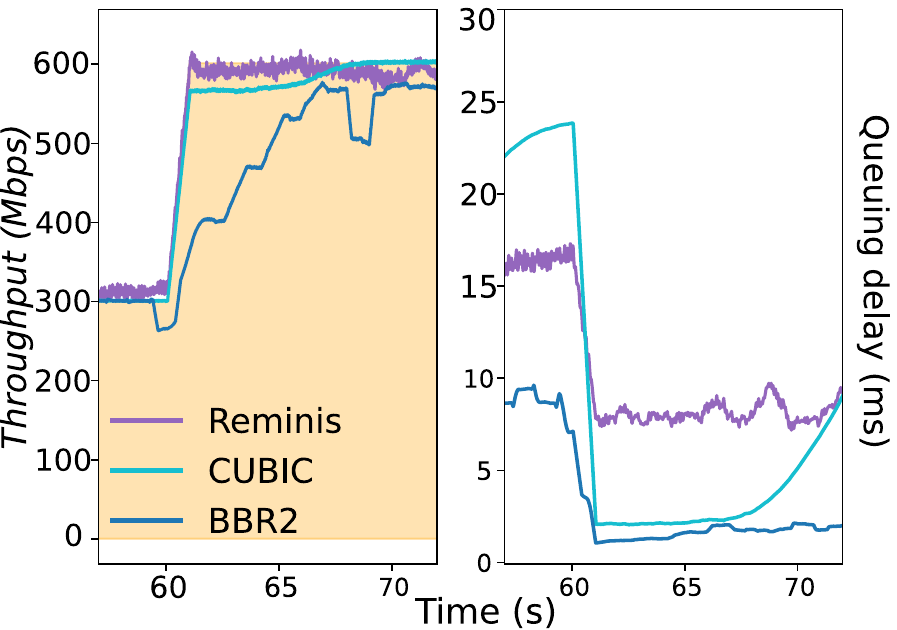}
        \subcaption{Step-Up}
        \label{fig::step_up}
    \end{minipage}
    \begin{minipage}[b]{0.24\linewidth}        
        \centering
        \includegraphics[width=0.95\linewidth, height=1.5in]{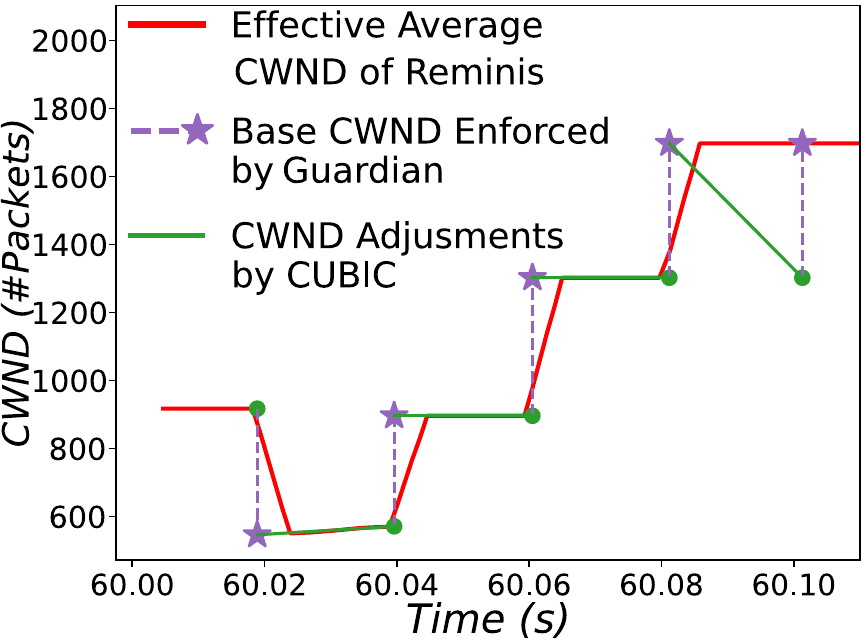}
        \subcaption{Congestion Win.}
        \label{cwnd_trans_up}
    \end{minipage}
    \begin{minipage}[b]{0.24\linewidth}        
        \centering
        \includegraphics[width=0.95\linewidth, height=1.5in]{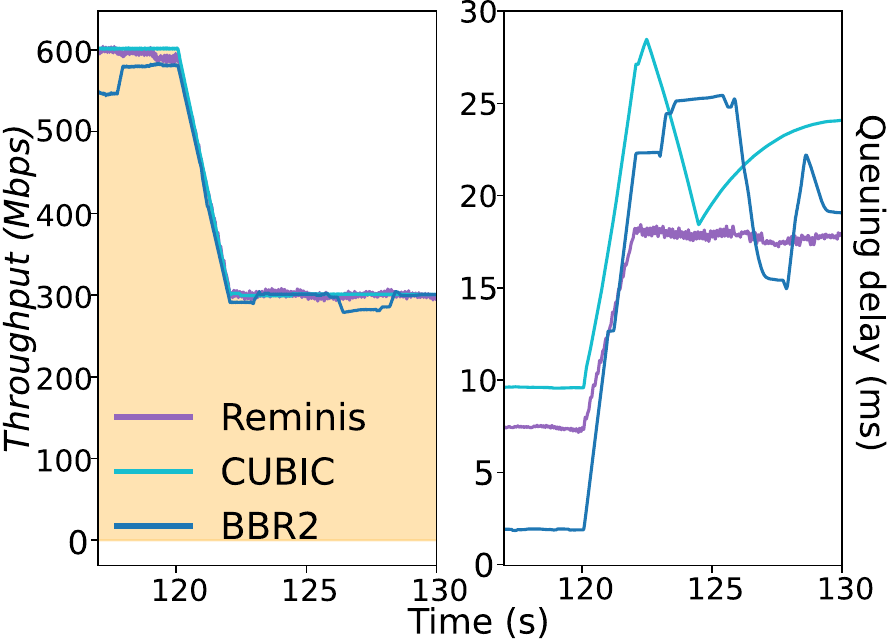}
         \subcaption{Step-Down}
        \label{step_down}
    \end{minipage}
    \begin{minipage}[b]{0.24\linewidth}        
        \centering
        \includegraphics[width=0.95\linewidth, height=1.5in]{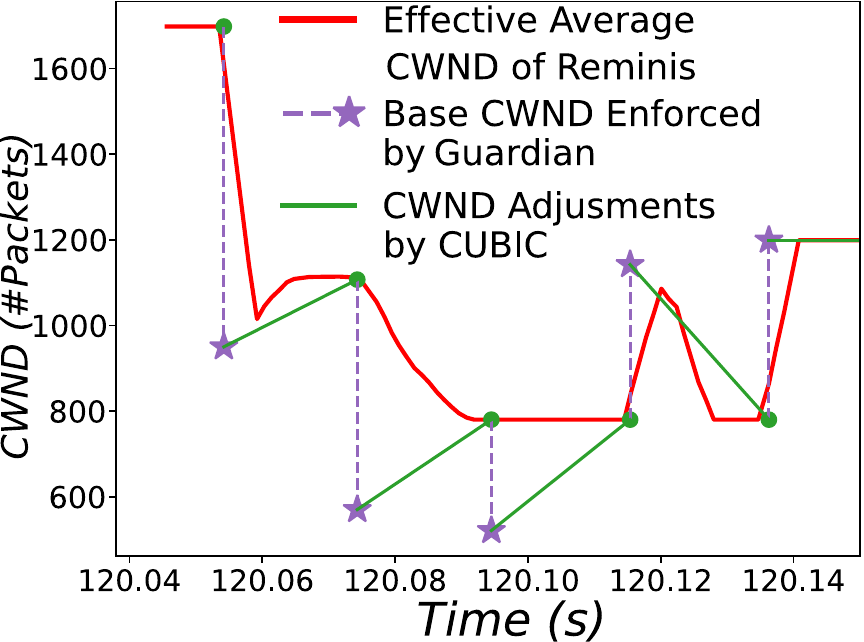}
         \subcaption{Congestion Win.}
        \label{cwnd_trans_down}
    \end{minipage}
    \caption{Dynamics of Reminis in a Step Scenario}
    \label{step}
\end{figure*}
\subsection{Dynamics of Reminis}
For showing the dynamics of Reminis, we use simple scenarios to depict the underlying actions of different blocks in Reminis. To put it in context, we also illustrate the performance of CUBIC and BBR2 here.  \par
\textbf{Reminis Response to Changes in Link Bandwidth:} 
To elaborate on the dynamics of Reminis, we use two different step scenarios in which we suddenly decrease/increase the link bandwidth. Probing the behavior of Reminis in a simple step scenario will help us understand Reminis responds to more complex traces as the channel capacity can be modeled as a summation of shifted step functions. \par
Fig. \ref{step} shows Reminis sending rate, queuing delay, and the CWND over a link with capacity changing from 300 Mbps to 600 Mbps and vice versa. The intrinsic RTT of the experiment is 20 ms and DTT is 40 ms.

In Fig. \ref{fig::step_up} when the link capacity increases, CUBIC and BBR2 are very slow to utilize this change. It takes them a few seconds to increase their sending rate to a point where they can utilize the link. Reminis, however, by inferring Zone 1, increases the CWND very fast. This increase in CWND value helps Reminis to fully utilize the link bandwidth shortly after the change in capacity. The Proactive Slowdown module is helpful here to stop Reminis from increasing the CWND too much. Fig. \ref{cwnd_trans_up} shows the Reminis CWND value during the increase in link capacity. This figure shows that the Non-Deterministic Exploration module increases the CWND at each SI and helps Reminis to adapt to this change quickly unlike BBR2 and CUBIC.\par
Fig. \ref{step_down} shows the performance of Reminis when link bandwidth decreases from 600 Mbps to 300 Mbps. This figure shows that CUBIC and BBR2 suffer from a surge in their delay while Reminis can control the delay increase so it will be still less than DTT. All the depicted CC schemes in Fig. \ref{step_down} adapt their sending rate to link bandwidth quickly but BBR2 and CUBIC have already occupied the queue so much that when the link capacity gets halved, they will experience a substantial surge in delay. Reminis is quick enough to adapt to the new scenario so the e2e delay does not exceed the delay target. Reminis will infer Zone 2 and 3, and as explained in Algorithm \ref{alg:1}, it will start decreasing the CWND.\\
Fig. \ref{cwnd_trans_down} shows changes in CWND with higher granularity. During this time, the Proactive Slowdown and the Catastrophe Mitigation modules decrease the CWND to match the new link capacity. After decreasing the CWND, if Reminis has decreased the window too much, the Non-Deterministic Exploration module will be activated (the last two samples of Fig. \ref{cwnd_trans_down}) and increases the CWND. Moreover, looking at the CWND adjustments done by the AIMD module in between each SI, it is obvious that in the first 3 SIs, as the AIMD module does not detect any packet loss, it keeps increasing the CWND between each SI.\\
\subsection{Impact of AIMD Block}
\label{micro::aimd}
Here, we investigate the answer to the question of why Reminis accompanies the Guardian block with a simple AIMD module. To observe the impact of the classic AIMD module, we do a simple ablation study and remove the AIMD block from Reminis. Using this new implementation of Reminis, we repeat all experiments in the SA scenario as described in Section~\ref{sec:eval:emulation} and gather the new version's overall throughput and delay performance. The results show that removing the AIMD block from Reminis leads to losing, on average, 15\% (and up to 30\%) link utilization without any tangible improvements in delay performance. To give more intuition about the effect of this module, Fig. \ref{AIMD_fig} shows a slice of a sample 5G trace. This figure illustrates how without the AIMD block, Reminis fails to keep up with increases in the access link capacity. These results demonstrate the role of the AIMD block in Reminis. The main intuition here is that since the Guardian works periodically (one action per SI), there are scenarios in which the Guardian can still miss the channel dynamics happening during one SI as 5G links can change on very small time scales. That is where an AIMD block comes into play. A simple AIMD block with its Ack-triggered logic adds fine-grained dynamics to the Reminis' actions and enables it to be more agile. 

Another benefit of the AIMD module is increasing the number of RTT samples during each SI, used to calculate the delay at each SI. Increasing the number of samples helps Reminis to get more reliable average statistics during each SI as averaging among more RTT samples reduces the measurement noise. In short, by providing more samples, AIMD helps NCI to have a better view of the network condition. 
\begin{figure}[!t]
    \centering
    \label{AIMD}
    \includegraphics[width=\linewidth]{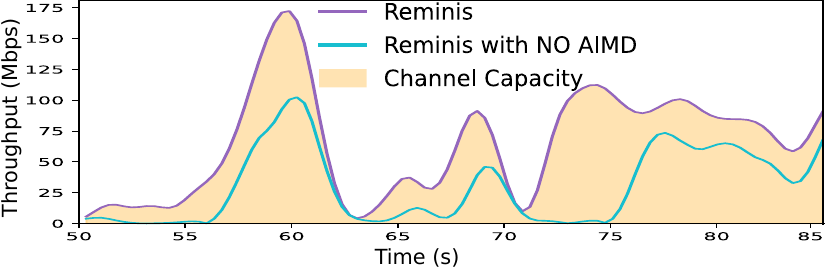}
     \caption{Impact of AIMD Block}
     \label{AIMD_fig}
\end{figure}


\begin{figure*}[!t]
    \centering
    \begin{minipage}[b]{0.33\linewidth}
        \centering
         \includegraphics[width=0.97\linewidth]{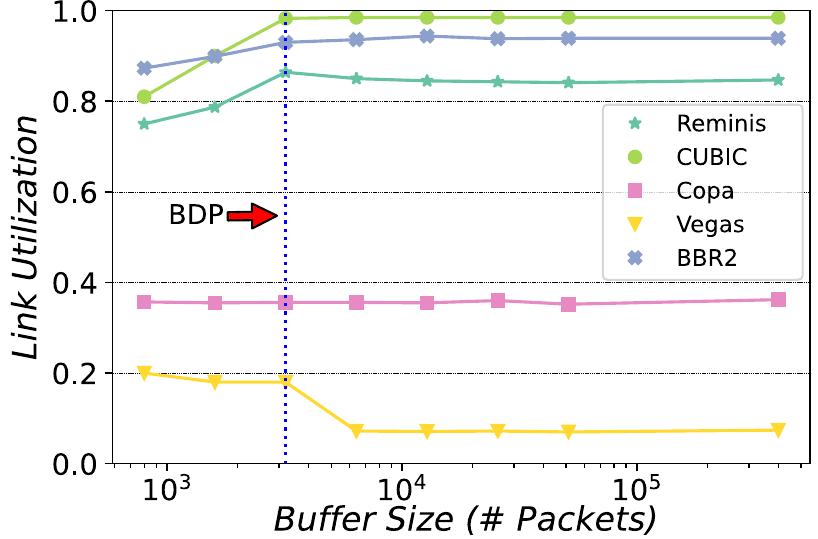}
        \caption{\label{buf-util} Buffer Size v. Throughput}
    \end{minipage}        
     \hfill
    \begin{minipage}[b]{0.33\linewidth}
        \centering
        \includegraphics[width=0.97\linewidth,height=1.45in]{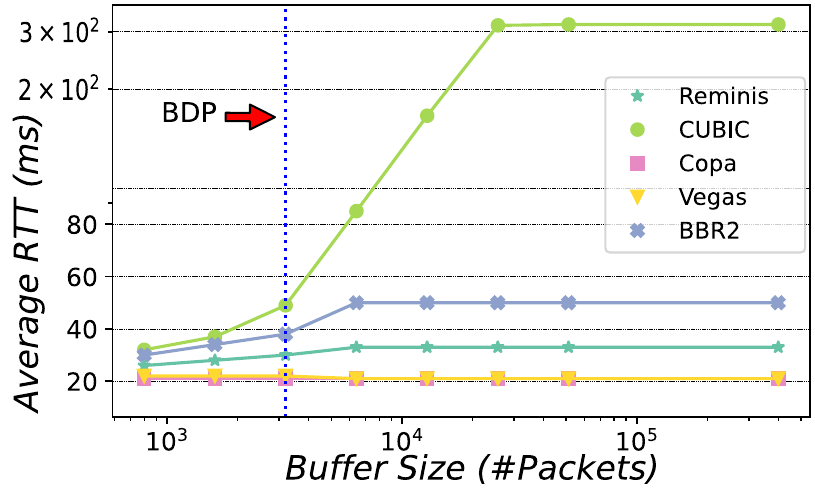}
        \caption{\label{buf-delay} Buffer Size v. Delay}
        \label{fig_geni}
    \end{minipage}
    \hfill
    \begin{minipage}[b]{0.33\linewidth}
        \centering
        \includegraphics[width=0.97\linewidth]{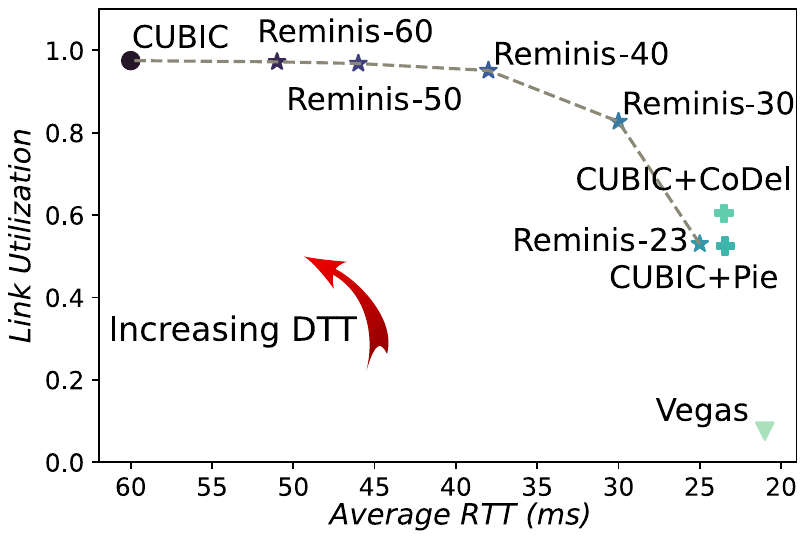}
        \caption{Impact of DTT}
        \label{DTT} 
    \end{minipage}

\end{figure*}

\subsection{Impact of Buffer Size}
One of the main characteristics of cellular networks is having per-user buffers at the base station. This helps the network to reduce the number of dropped packets and to offer a more reliable network to the users. Having separate queues for each user means that users don't compete over a common queue. This feature, despite mitigating some issues like fairness between multiple users' flows, leads to the well-known problem of bufferbloat \cite{bufferbloat2} and self-inflicted delay \cite{sprout}.\par
In this section, to measure the impact of different bottleneck buffer sizes on the performance of different schemes, we change the buffer size of the emulated network from 800 packets to 51200 packets. The choice of the lowest buffer size, 800 packets, is based on the findings of prior work \cite{xu2020understanding} regarding the buffer size of the NSA-5G network.\\
As expected and shown in Fig. \ref{buf-delay}, CUBIC tries to occupy all available buffer. This approach means that with increasing the buffer size, CUBIC's delay performance degrades. In contrast, the average delay performance of Reminis is roughly independent of buffer size and is around the value of DTT (30 ms). This behavior from Reminis is rather expected as it tries to control the CWND so that the overall delay meets the DTT requirement. Moreover, Fig. \ref{buf-util} shows that despite keeping the delay constant, Reminis achieves around 80\% link utilization regardless of the underlying buffer size.\par
NSA-5G networks have smaller buffers (4G buffers), so as long as there are still NSA base stations in the network, any proposed 5G CC algorithm should be able to also have a good performance in low buffer settings. Meeting its DTT, Reminis, can achieve 80\% link utilization in NSA buffer size.

\subsection{Impact of Delay Tolerance Threshold}
\label{DTT_effect}
Delay Tolerance Threshold (DTT) is a key parameter in Reminis' design and performance. Reminis will become more conservative when the measured delay exceeds DTT, so we expect to have a trade-off between delay and link utilization based on different values of DTT. Large DTT values steer Reminis toward being more throughput-oriented, while small DTT values guide Reminis toward being more delay-oriented. In Fig. \ref{DTT}, Reminis-X means a version of Reminis where the DTT parameter has been set to X ms.

For comparison, TCP Vegas, one of the major delay-oriented CC schemes, and TCP CUBIC, the most throughput-hungry CC scheme, have been added to Fig. \ref{DTT}. In addition, we accompany CUBIC with two active queue management (AQM) schemes, CoDel \cite{codel} and Pie \cite{pie}. Although these schemes make changes in the network, they still cannot utilize the link more than 60\%, which shows a major drawback for these AQM schemes on 5G links.
Fig. \ref{DTT} shows that DTT has the expected impact on the performance of DTT. With a larger DTT, we can guide Reminis to be more throughput-hungry and a smaller DTT makes Reminis more delay-sensitive. One salient point in this experiment is that Reminis does not compromise an immense amount of throughput to meet its DTT. For instance, for DTT=30 ms, Reminis achieves its goal while its link utilization is only reduced by 20\% compared to CUBIC which has an average RTT equal to 60 ms.
\begin{figure*}[!t]
    \begin{minipage}[b]{0.33\linewidth}          
        \centering
        \includegraphics[width=0.85\linewidth, height=1in]{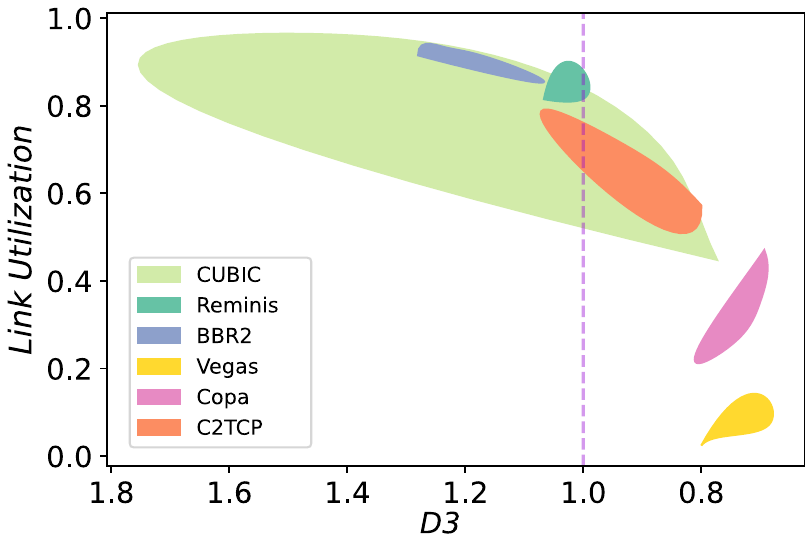}
        \caption{Impact of min. delay}
        \label{mRTT}
    \end{minipage}
    \begin{minipage}[b]{0.33\linewidth}        
        \centering
        \includegraphics[width=0.85\linewidth, height=1in]{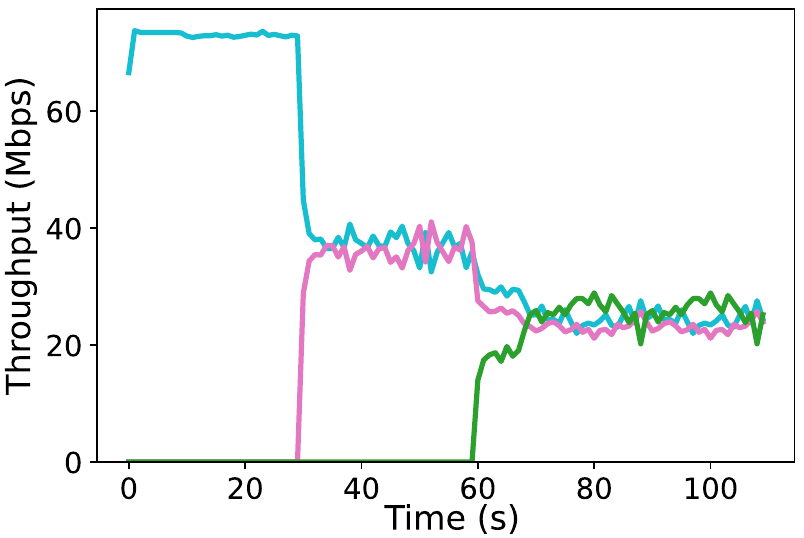}
        \caption{Reminis' Fairness}
        \label{fairness}
    \end{minipage}
    \begin{minipage}[b]{0.33\linewidth}        
        \centering
        \includegraphics[width=0.85\linewidth, height=1in]{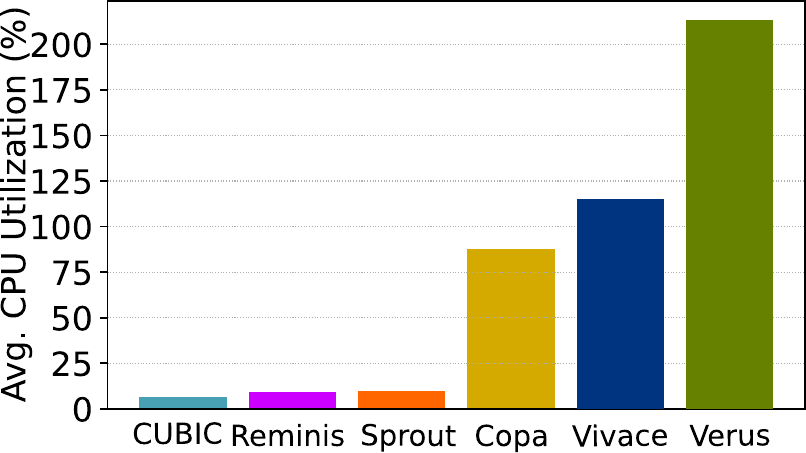}
        \caption{CPU Utilization}
        \label{cpu_util}
    \end{minipage}
\end{figure*}

\subsection{Impact of Network's Intrinsic RTT}
The intrinsic RTT of a network is a function of different things including the UE-Server distance \cite{narayanan2021variegated}. Therefore, in this experiment, we evaluate the performance of different schemes for networks with different intrinsic RTTs. In particular, we change the intrinsic RTT values of our emulated networks to $\{5,10,20,30,40,50\}\:ms$ and set corresponding DTT values to $\{7.5, 15, 30, 45, 60, 75\}\:ms$. 
We assume a SA version and consequently, set the buffer size to BDP for each tested intrinsic RTT value. 
\footnote{Note that Reminis automatically adjusts the value of SI to the observed mRTT of the network. That means, by design, Reminis utilizes different SI values for different settings.}
\color{black}
We define deviation from the desired delay (D3) parameter as $D3 = \frac{d}{DTT}$. 
D3 index indicates how much the average delay ($d$) is larger/smaller than DTT with D3=1 meaning the average delay has met DTT.
\par
Fig. \ref{mRTT} depicts the convex hull of each compared CC scheme in all the possible scenarios based on different values of intrinsic RTT. For each intrinsic RTT, we repeat the experiment 3 times. As Fig.~\ref{mRTT} shows, Reminis has a D3 less than 1.03 and roughly 80\% link utilization over the tested range of intrinsic RTTs. On the other hand, delay-based approaches like Copa or Vegas, show poor utilization performances. For instance, Copa's link utilization can go down to 20\%. Moreover, C2TCP can achieve a good D3, but its utilization can vary widely over different tested intrinsic RTTs which is not desirable. For BBR2, D3, and for CUBIC both D3 and link utilization have high variations based on the value of mRTT which is not a desirable feature for cellular networks.
\subsection{Fairness}
Here, we investigate the fairness property of Reminis. We have created a network containing servers and one client connected via an emulated bottleneck link. We set the intrinsic RTT of the network to 20 ms. We send three separate Reminis flows toward the client, with 30 seconds gaps between the start of each flow, and measure the throughput of each flow at the client side. Fig.~\ref{fairness} shows the result of this experiment and demonstrates that Reminis flows can fairly share the bottleneck. When the second Reminis flow enters the network, at around $t=30$s, the first flow detects Zone 2/3, hence it reduces its CWND. This will release enough bandwidth for the second flow and consequently, two flows can share the bottleneck bandwidth. The same happens when the third flow enters the network.

\subsection{CPU Overhead}
Considering the power constraints of devices in 5G networks, it is crystal clear that a successful 5G-tailored CC scheme should be lightweight with low computational overheads. That said, here we show the lightweight aspect of Reminis. To that end, 
\color{black}
 We measure the average CPU utilization of different CC schemes by sending traffic from a server to a client over a 720 Mbps emulated link for 2 minutes. The choice of 720 Mbps for link capacity comes from the average of 5G link capacity measured in prior work \cite{narayanan2020lumos5g}. Fig. \ref{cpu_util} shows the comparison between the average CPU utilization of different schemes. Even though the current version of the Guardian module is implemented in user space, Reminis shows a very good performance in terms of CPU utilization. In particular, CPU utilizations of Reminis and CUBIC (the default CC scheme implemented in Kernel) are around 9.2\% and 6\%, respectively. 
 Having low overhead is a consequence of keeping the Reminis design simple, employing simple delay statistics, and utilizing low-overhead ack-triggered (AIMD) actions.\color{black}

\section{Discussion}
\noindent{\textbf{Choice of DTT:}} If the application is delay sensitive and can provide Reminis a DTT parameter, Reminis can use this value. Reminis starts after the slow-start phase of the AIMD block and therefore, it will have enough samples to infer the mRTT of the network. However, if the application's requested DTT is less than mRTT, Reminis can detect this problem and enforce DTT to be larger than the measured mRTT of the network. If the applications do not provide a specific DTT, Reminis switches to its default where$DTT = 1.5mRTT$. \\
\noindent{\textbf{Does Reminis guarantee to always meet DTT requirements?}} A novel concern in 5G networks is the occasional 5G ``dead zones''. As explained in Section~\ref{sec:intro:5g_is_diff}, users entering these zones experience close or equal to zero link capacity which can last for seconds. As no packet is served from the queue during this time, the queuing delay will inevitably soar. Therefore, any e2e CC scheme by nature, including Reminis, can not control these types of scenarios. The bottom line is that although we showed Reminis is significantly better at controlling e2e delay in these scenarios than other state-of-the-art CC schemes, any e2e CC algorithm can only deliver QoS demands that are feasible in a network.\par
\noindent{\textbf{Limitations of Reminis:}} 
Reminis targets emerging applications with low/ultra-low latency requirements. Supporting such applications before anything requires networks with low/ultra-low intrinsic delays. This justifies and, indeed, encourages the use of edge-based architectures such as MEC. In such settings, competing with loss-based CC schemes like CUBIC, which fully utilize buffers, is less of a concern. That said, when these settings are not held and Reminis coexists with loss-based flows that fully occupy queues, similar to any CC scheme that attempts to control the e2e delay, it faces problems when the DTT value (and mRTT) is way lower than the queuing delay caused by loss-based flows.

\section{Conclusion}
In this work, we demonstrate that achieving high throughput and low controlled delay in highly variable and unpredictable 5G networks does not necessarily require convoluted learning-based schemes or prediction algorithms.
To that end, we introduce Reminis, a simple yet adaptable e2e CC design tailored for 5G networks with provable convergence properties, and we show that properly exploiting non-deterministic throughput exploration algorithms combined with proactive/reactive delay control mechanisms is sufficient to effectively adapt to 5G cellular networks and achieve high performance. 
Our controlled emulations and real-world experiments, show the success of Reminis in acquiring its design goals and demonstrate that Reminis can outperform state-of-the-art CC schemes on 5G networks while being deployment friendly with low overheads and not requiring any changes in cellular network devices.\color{black}

\bibliographystyle{ACM-Reference-Format}
\bibliography{main}

\begin{appendices}
\section{Related Works}
\label{sec:app:related}
\noindent\textbf{General Purpose CC Schemes:}
Since the early days of congestion control designs, many different schemes have been proposed. TCP Reno \cite{reno}, TCP NewReno \cite{newreno}, TCP CUBIC \cite{cubic} and BIC \cite{bic} are among important legacy loss-based CC schemes. Delay-based TCP variants like TCP Vegas \cite{vegas} were also introduced to use the delay as the main congestion signal. One of the basic assumptions of all these legacy CC protocols is to have a fixed-capacity link, which does not apply to cellular networks. Among recent works in this category, Copa \cite{copa} and LEDBAT \cite{ledbat} can be named. \par


\noindent\textbf{Environment-Aware CC schemes:} Many prior works focus on cellular networks specifically. For instance, Sprout \cite{sprout} employs packet arrival times combined with a probabilistic inference to make a cautious forecast of packet deliveries in cellular networks.
We found that Sprout is inefficient for 5G links and its model of the network path does not hold in 5G networks. We even tried to update its model by changing its hard-coded parameters, but even after that, sprout performed poorly in 5G networks. Another work is Verus~\cite{verus} which tries to make a delay profile of the network and then use it to calculate the CWND. However, as we showed in section~\ref{sec:eval}, Verus is not agile enough to adapt to 5G networks. 
Among more recent cellular-tailored works C2TCP~\cite{c2tcp} and ExLL~\cite{exll} can be named. C2TCP works on top of classic throughput-oriented TCP and accommodates various target delays. However, when running on 5G links, C2TCP cannot adapt itself to quick surges of 5G links leading to under utilization.
ABC \cite{abc2} 
uses routers' feedback to help senders adjust their CWND value in a wireless setting. However, the feasibility of deploying in-network feedback approaches that require either new devices or an update of the existing ones remains a major obstacle for deploying these schemes that are not fully e2e, in particular in cellular networks where network providers are not very flexible in terms of changing their equipment.
Another group of cellular-tailored CC schemes focuses on physical layer measurements like PBE-CC \cite{pbe}, CQIC \cite{cqic} and CLAW \cite{claw}. General concerns around these works are their power consumption and privacy as they need to access end hosts' physical layer information. Moreover, considering CU/DU splitting in 5G \cite{cu-du}, the applicability of these schemes should be justified.\par

\noindent\textbf{ML-based CC Schemes:}
Remy~\cite{remy} uses offline learning to map network events to CC-based actions in a brute-force manner. However, Remy's approach needs accurate assumptions about the network, which is not applicable to highly variable 5G networks. PCC-Vivace~\cite{vivace} leverages online learning to choose the best sending rates, but the online-learning approach used in Vivace makes it very slow to react to quick changes of capacity in 5G networks (e.g., see section~\ref{others_fail}).
Some other ML-based schemes target cellular  networks. DeepCC~\cite{deepcc} is a cellular-tailored plugin that utilizes deep reinforcement learning to control the maximum congesting window of an underlying scheme such as CUBIC. However, it cannot control e2e delay very well in a 5G setting (e.g., see section \ref{sec:eval}). The reason likely lies in the fact that DeepCC's model, which is trained over 4G links, cannot generalize well to unseen 5G networks. 
Besides the generalization issue, another concern with ML-based schemes is their interpretability. Although the goals of these models are clear, they are usually treated as black boxes. In contrast, heuristic models, such as Reminis, can be modeled and analyzed rigorously (as is done in our work). Therefore, debugging ML-based CC schemes to generalize them to new environments is difficult and might require making fundamental changes in their designs. Although some recent works like \cite{interpreting}  have focused on interpreting ML models used in the networking field, this line of work is in its early stages.

\begin{figure*}[!t]
    \centering
    \begin{minipage}[b]{.33\linewidth}
        \centering
        \includegraphics[width=0.97\linewidth]{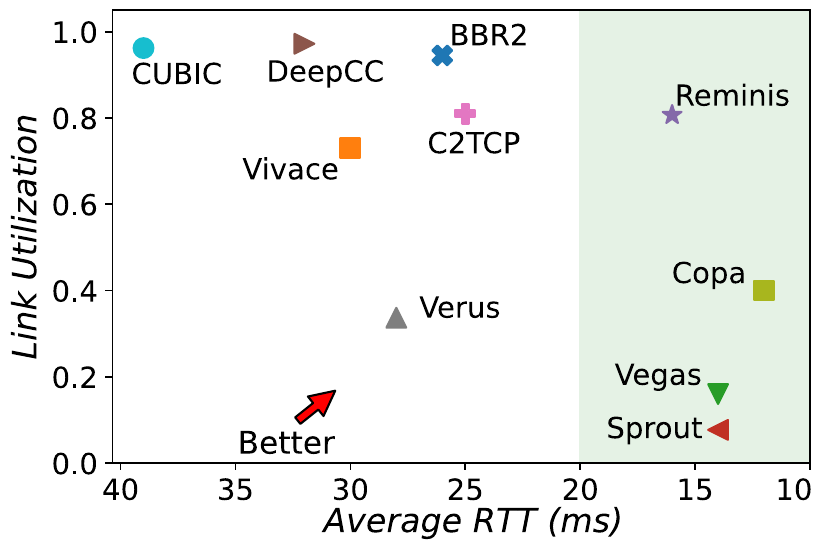}
        \caption{MEC Experiment}
        \label{mec_avg}
    \end{minipage}
    \begin{minipage}[b]{0.66\linewidth}
        \begin{minipage}[b]{0.49\linewidth}
            \centering
            \begin{center}
        \includegraphics[width=0.97\linewidth]{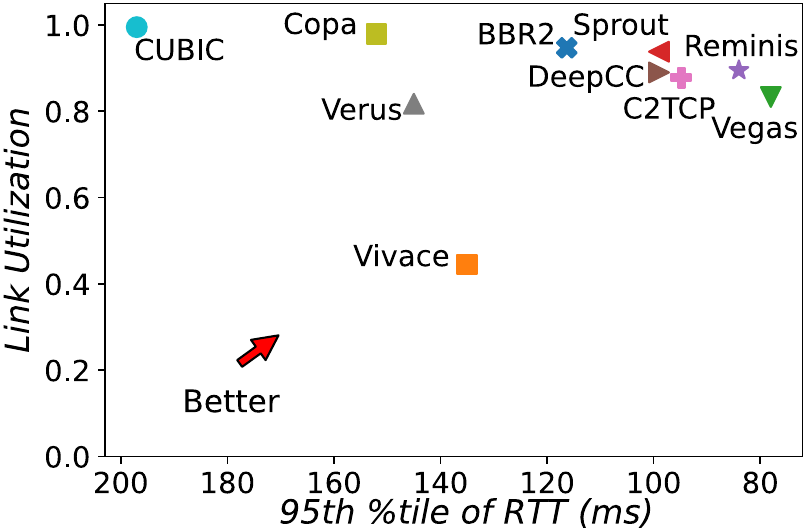}
        \end{center}
        \end{minipage}
        \hfill
        \begin{minipage}[b]{0.49\linewidth}
            \centering
            \includegraphics[width=0.97\linewidth]{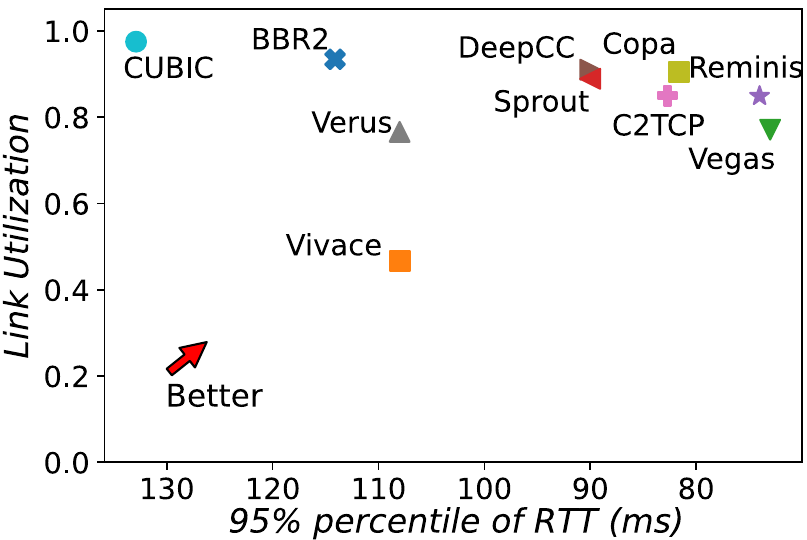}
        \end{minipage}
        \caption{Throughput-Delay Performance on 3G (left) and 4G (right) Traces}
        \label{3G_4G_eval_app}
    \end{minipage}
\end{figure*}
\section{MEC-Flavored Emulations}
\label{app::MEC}
As explained in Section \ref{sec::MEC}, mobile edge computing integrated with 5G is one of the design opportunities in 5G networks. 5G aims to support under 10 ms latencies using the New Radio technology so it can support applications like AR/VR. Our emulations in Section \ref{sec::MEC} show that Reminis can meet delay requirements of AR/VR applications in a MEC environment. In Section \ref{sec::MEC} we showed the throughput-95th percentile of delay performance of Reminis and here Fig. \ref{mec_avg} will show the throughput-average delay performance of Reminis and other competing CC schemes. As Fig. \ref{mec_avg} shows, considering average delay as the metric, Reminis is capable of meeting the delay requirement of an AR/VR application while keeping the throughput around 80\%, at least 2.4 $\times$ more than other schemes which meet the delay requirement ( in the green area).

\section{Reminis Performance on Legacy
Cellular Networks}
\label{app::3G/4G}
As discussed in Section \ref{3G/4G}, Reminis is designed to contend with access link volatility hence the performance benefits of Reminis are not tied to only 5G networks and it still can be the performance frontier over other networks like 3G/4G. Fig. \ref{3G_4G_eval_app} shows the throughput-95th percentile of delay performance of Reminis and other competing CC schemes on 3G and 4G traces gathered by \cite{deepcc} and \cite{c2tcp} respectively. Based on Fig. \ref{3G_4G_eval_app}, on 3G and 4G networks, compared to BBR2, Reminis can decrease the 95th percentile of RTT by 1.475 $\times$, and 1.485 $\times$, while BBR2's throughput is only 1.1$\times$ and 1.05$\times$ more than Reminis respectively.

\section{Mahimahi Patch}
\label{appendix:MM_patch}
As mentioned in Section \ref{sec:eval:emulation}, some designs decision in Mahimahi hinders its ability to emulate high-speed links. Mahimahi logs three different events which are channel tokens( emulating the link capacity), packet enqueues, and packet dequeues. In the default code of Mahimahi, these three events are logged using a single thread function. However, this design decision is not scalable to high-speed 5G links. Dedicating one thread for each of the mentioned events, solved the problem, and the patched Mahimahi was not limited anymore. Fig. \ref{appendix:mahi} depicts the difference between the default version of Mahimahi and our patched version. We use TCP CUBIC to show the difference between these two versions. As shown in Fig. \ref{appendix:mahi}, default Mahimahi is hindered in emulating high throughput scenarios, however, with the mentioned modifications, we could solve this problem.
\begin{figure}[h]
    \centering
     \includegraphics[width=0.9\linewidth]{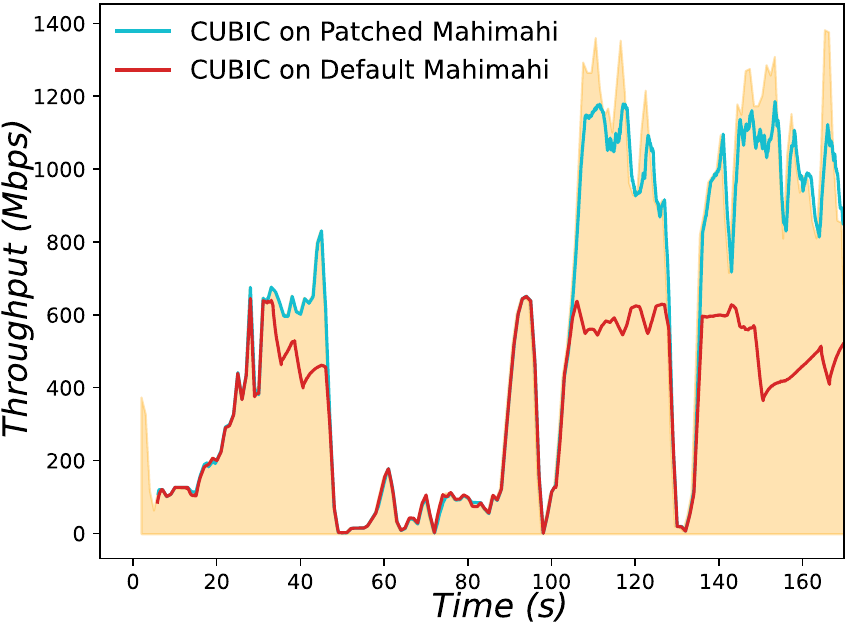}  
    \caption{Mahimahi Patch}
    \label{appendix:mahi}
\end{figure}


\section{Convergence Proof}
\label{conv_proof}

\begin{figure}[!t]
\includegraphics[width=\linewidth]{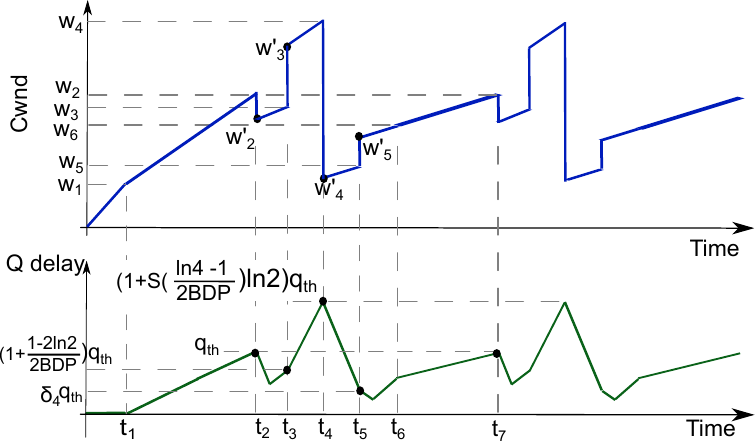}
\caption{Reminis Convergence Analysis}
\label{analysis}
\end{figure}
In this section, we analyze Reminis convergence properties over a fixed bandwidth link with an infinite buffer. We show that Reminis, on average, converges to a steady state and provide an upper bound on queuing delay caused by Reminis in this scenario.\\
Let us start by describing our assumptions, and notation: 
\begin{enumerate}
    \item Reminis DTT is $2\times mRTT$. This means that the  queuing delay that Reminis targets is $mRTT$. From now on we will denote this target queuing delay with: $q_{th} = mRTT$.
    \item BDP is $BW\times mRTT$ where $BW$ is the link capacity in packets per second.
    \item $w_1$ is the CWND that fully utilizes the link without causing any queuing delay i.e. $w_1=BDP$.
    \item $w_2$ is the CWND fully utilizing the link and creating a total e2e delay of $2mRTT$. ($w_2 =2w_1$)
     \item $w'_n$ is Reminis CWND at $SI_n$ after being modified by the PE, whereas $w_n$ is the cwnd before applying PE's changes.
     \item We assume that the AIMD block is in the congestion avoidance phase.
     \item We assume that the intrinsic RTT of the network is equal to the mRTT.
 \end{enumerate}

In this section, we prove the following main Theorem. Note that $S$ is the Sigmoid function introduced in Section \ref{desing::NDE}.
\begin{theorem}
\label{main_th}
On average, Reminis converges to a steady state with a queuing delay no more than $(1+S(\frac{\ln{4}-1}{2BDP})\ln{2})q_{th}$.
\end{theorem} 
\begin{proof}
We begin with two lemmas which will come in handy in the proof.
\begin{lemma}
\label{sig_exp}
The expected value of Sigmoid function where its variable is under a Normal distribution ($S(x):x\sim\mathcal{N}(\mu,\,\sigma^{2})$)is:
\begin{equation*}
    \EX[S(x)] = S(\frac{\mu}{\sqrt{1+\pi \frac{\sigma^2}{8}}})
\end{equation*}
\end{lemma}

\begin{lemma}
\label{taylor}
The Taylor series of $2^x$ for x around zero is equal to: $2^x \approx 1+ x \ln{2}+ x^2 \frac{{(\ln{2})}^2}{2!}+ x^3 \frac{{(\ln{2})}^3}{3!} + ...$
\end{lemma}
\par
To prove Theorem \ref{main_th}, we break down the time axis into smaller intervals and investigate what happens in each of them. \par
\textbf{1:} {{$\pmb 0\pmb \le \pmb t \pmb \le \pmb t_1$}}: $t_1$ is the time where Reminis CWND reaches $w_1$ for the first time. If running without Reminis, the underlying AIMD block will increase the cwnd by one every $mRTT$ which means it would take $w_1\times mRTT$ for the AIMD block alone to reach $w_1$. The behavior or Reminis during $0\le t\le t_1$ could be summarized by the following Theorem.
\begin{theorem}
Reminis helps AIMD logic to reach $w_1$ in $\mathcal{O}(\log{w_1})$ instead of $\mathcal{O}(w_1)$.\\
\end{theorem}
The proof of the above Theorem is presented in Section \ref{Speed_up_proof} of the Appendix. \par
\textbf{2:} \textbf{{$\pmb t_1 \pmb \le \pmb t \pmb < \pmb t_2$}}: Now, let us assume that $t_2$ is the first time that the cwnd reaches $w_2$.
After reaching $w_1$ at $t_1$, the AIMD block keeps increasing the cwnd by one packet each $mRTT$. From this point on, the queue starts to increase gradually, but the overall delay is still less than DTT. Also, as delay derivative is positive during $t_1 <t < t_2$, Reminis will infer Zone 2. Up until $t=t_2$, where the cwnd reaches $w_2$, the Guardian will not make any extra changes in the cwnd.
\begin{equation}
    t_2 - t_1 = (w_2 - w_1)\times mRTT  = w_1\times mRTT
\end{equation}
As we mentioned earlier, the delay will increase during $t_1 < t < t_2$ so during this time period the derivative of delay is:
\begin{equation}
\label{der_T2}
    \nabla d_{t_1 \le t \le t2}  = \frac{q_{th}}{t_2 - t_1}= \frac{1}{w_1}
\end{equation} \\
During $t_1\le t\le t_2$, there will be $\frac{w_1 \times mRTT}{mRTT} =  w_1$ SIs, as SIs happen once per $mRTT$. In each of these SIs, Reminis will measure a delay derivative calculated in Equation \ref{der_T2}. \\
Based on the update rule for the Gaussian distribution, after reaching $t_2$, the mean of the Gaussian distribution will be: $1- \sum_{n=1}^{w_1} \frac{1}{w_1} = 0$.

\textbf{3:} \textbf{\texttt{$\pmb t \pmb = \pmb t_2$}}: Proactive Slowdown module will be activated in this SI as the estimated delay for next SI, based on Equation \ref{prediction}, exceeds DTT. Based on Algorithm \ref{alg:2}, Proactive Slowdown module will decrease the cwnd and set it to be $w'_2 = w_2 \times 2^\frac{-1}{w_1}$. Using Lemma \ref{taylor}, we can further simplify the equation and have $w'_2 = w_2 \times (1- \frac{\ln{2}}{w_1})$. For simplicity, we call $\delta_1=\frac{\ln{2}}{w_1}$. Therefore, we will have: $w'_2 = w_2 (1-\delta_1)$

\textbf{4:} \textbf{\texttt{$\pmb t_2 \pmb < \pmb t \pmb < \pmb t_3 \pmb = \pmb t_2 \pmb + \pmb m \pmb R \pmb T \pmb T$}}: After changes made by Proactive Slowdown module at $t=t_2$, the queue will start draining. During $t_2 <t < t_3 $ AIMD adds one packet to cwnd so the cwnd at $t=t_3$ will be:
\begin{equation}
\begin{split}
    w_3 & = w'_2 +1 = w_2 (1-\delta_1) +1\\
    & = w_2 + (1- 2\ln{2}) = w_2 + \delta_2 
    \end{split}
\end{equation}
As $\delta_2$ is negative, we conclude: $ w'_2<w_3<w_2$. An important point during this time interval is that starting from $t = t_2$, as the cwnd is less than the number of in-flight packets, Reminis does not send any packets and the queue starts depleting. However, at $t_c = t_2+ \frac{2\ln{2}}{BW} < t_3$ the number of in-flight packets become equal to cwnd hence Reminis starts sending packets again.

Considering the queuing delay graph, we can say that at $t_3$, the queuing delay is: $q_{t_3} = (1+ \frac{\delta_2}{w_2}) q_{th} = (1 + \delta_3)q_{th}$ where $\delta_3 = \frac{\delta_2}{w_2}$.
\\
Between $t_2$ to $t_3$ the gradient of delay is negative and thus the mean of the Gaussian distribution at $t_3$ would be:\\ $\mu_3 = -\frac{\delta_3 q_{th}}{minRTT} = -\delta_3$. Note that $\delta_3$ is negative.\par
\textbf{5:} $\pmb t \pmb = \pmb t_3$. The NCI module would infer Zone 1 and hence the Non-Deterministic Exploration module will be activated by the PE. Considering $w_2\gg 1$, we have: 
\begin{equation}
\label{w3}
\EX[w'_3]=\EX[w_3 \times 2^{s(x \sim \mathcal{N}(-\delta_3,\frac{-\delta_3}{4}))}]
\end{equation}
\\ We can further simplify Equation \ref{w3} by using Lemma \ref{taylor} and Lemma \ref{sig_exp}.
\begin{equation}
\begin{split}
    w'_3 &= w_3 \times \EX[2^{S(x)})] \approx w_3 \times \EX[1+S(x) \ln{2}] \\
    &\approx w_3 \times [1+ \ln2 \times S(\frac{-\delta_3}{\sqrt{1+ \pi \frac{{\sigma_3}^2}{8}}})]  = w_3 \times (1+\delta_4)
    \end{split}
\end{equation}
To be more precise, $t=t_3$ is the time when Reminis tries to explore larger cwnd values to increase throughput. It can be shown that $w'_3 > w_2$. \par

\textbf{6:} \textbf{\texttt{$\pmb t_3 \pmb < \pmb t \pmb < \pmb t_4 \pmb = \pmb t_3 \pmb + \pmb m \pmb R \pmb T\pmb T$}}: During this time the AIMD module adds one packet to the cwnd so we will have $w_4 =w'_3 +1$. The queuing delay at $t_4$ approximately would be: $d_4 = (1+\delta_4)q_{th}$. \\
At $t=t_4$, the mean of the Gaussian distribution will be $\mu_4 = -\delta_4$. The mean of the Gaussian distribution is negative as Reminis has just had an unsuccessful throughput exploration, which makes it more conservative for further explorations. \par
\textbf{7:} \textbf{\texttt{$\pmb t\pmb =\pmb t_4$}}: As $d_4$ is more than DTT, the Catastrophe Mitigation module is activated. We have: \texttt{SafeZone}($d_4$)=$- \delta_4$.\\ 
So, using Lemma \ref{taylor}, we can say we have:
\begin{equation}
        w'_4 = w_4 \times 0.5 \times 2^{-\delta_4}= w_1 (1+ \delta_4 (1-\ln{2}))
\end{equation}
Hence, at $t=t_4$, the Guardian reduces the cwnd to be a bit more than $w_1$. With this modification by the Catastrophe Mitigation module, the number of in-flight packets becomes bigger than the cwnd value.\par
\textbf{8:} \textbf{\texttt{$\pmb t_4 \pmb <\pmb t \pmb < \pmb t_5\pmb =\pmb t_4 \pmb + \pmb m\pmb R\pmb T\pmb T$}}: Up until $t=t_5$, the number of in-flight packets is more than cwnd. So no additional packet is sent and the queue keeps depleting. However, the underlying AIMD keeps adding one packet to the cwnd every mRTT. Therefore, we will have: $w_5 = w_4 +1$.\\
At $t_5$, we have: $q_5 =\delta_4 q_{th}$. Based on this value, the mean of Gaussian distribution will be updated to $\mu_5 = 1 - \delta_4$. \par

\textbf{9:} \textbf{\texttt{$\pmb t\pmb =\pmb t_5$}}: At $t_5$, the Non-Deterministic Exploration module gets activated as the NCI module detects Zone 1. Based on Lemma \ref{sig_exp} and \ref{taylor}, we can calculate $\EX[w'_5]$:
\begin{equation}
        \EX[w'_5] = w_5 [1+ \ln{2} \times S (\frac{1-\delta_4}{\sqrt{1+\pi \frac{(1-\delta_4)^2}{8}}})] = w_5 (1+\delta_5)
\end{equation}
Even with this increase in the cwnd, the number of in-flight packets is still more than the cwnd value as $w_4 - w_1$ which is the number of in-flight packets is bigger than $w_5$. However, it can be easily shown that $w_5 > w_4 - 2\times w_1$. Hence, we can conclude at some point in the next SI, the cwnd will be equal to the number of in-flight packets. \par

\textbf{10:} \textbf{\texttt{$t_5< t < t_6 = t_5 + mRTT$}}: At $t_6$, the cwnd will be $w_6 =w'_5 +1 $. Before reaching $t_6$, as the number of in-flight packets gets less than the cwnd, Reminis starts sending packets. So, at $t=t_6$, the queuing delay will be: $q_6 = q_{th} \times \frac{1+ \delta_5+\delta_4(1-\ln{2})}{2}$. Moreover, at $t_6$ the mean of Gaussian distribution will be $\mu_6 =\frac{1-\delta_5-\delta_4(1-\ln{2})}{2}$.

After $t_6$, the NCI module will infer Zone 2 and hence the Guardian will not change the cwnd value. This event is repeated in every SI until the cwnd reaches $w_2$ again at $t_7$. At this moment, the Gaussian mean has been reset to 0. So, this moment is similar to $t_2$. Hence, we can conclude that on average, Reminis performs similarly to $[t_2,t_7]$ periodically, and that leads to its steady state behavior.
\end{proof}
\begin{figure*}[!t]
\centering
\begin{subfigure}{.33\textwidth}
  \centering
  \includegraphics[width=0.9\linewidth,height=1.5in]{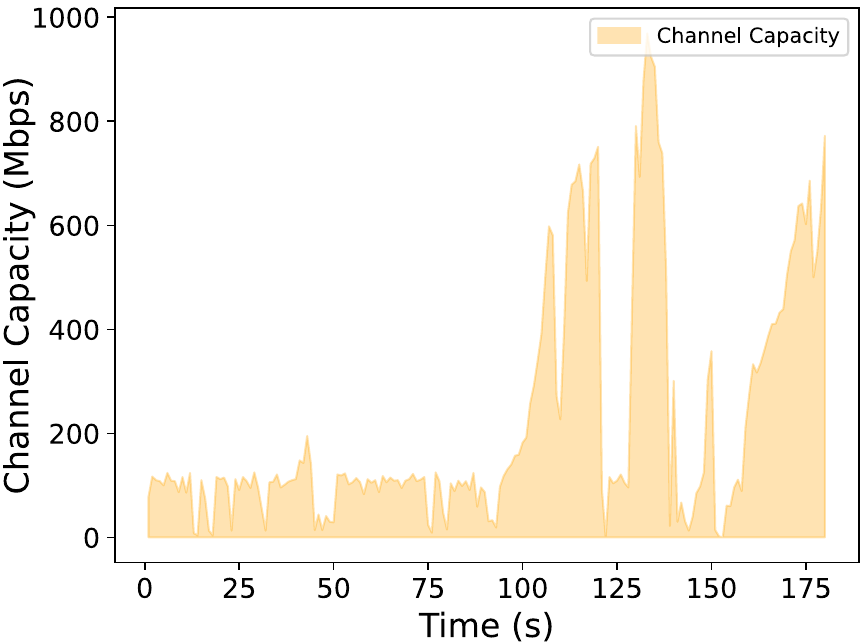}
\end{subfigure}%
\begin{subfigure}{.33\textwidth}
  \centering
  \includegraphics[width=0.9\linewidth,height=1.5in]{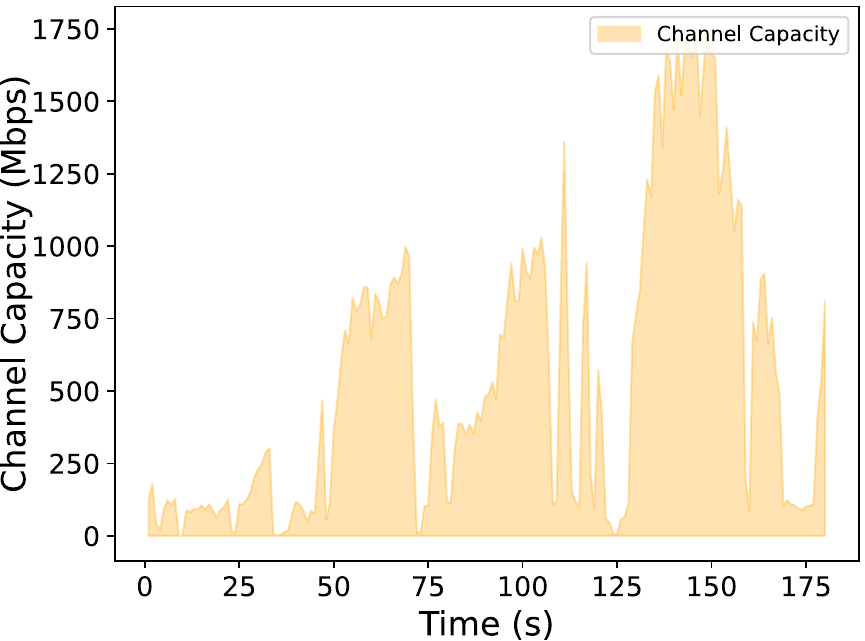}
\end{subfigure}%
\begin{subfigure}{.33\textwidth}
  \centering
  \includegraphics[width=0.9\linewidth,height=1.5in]{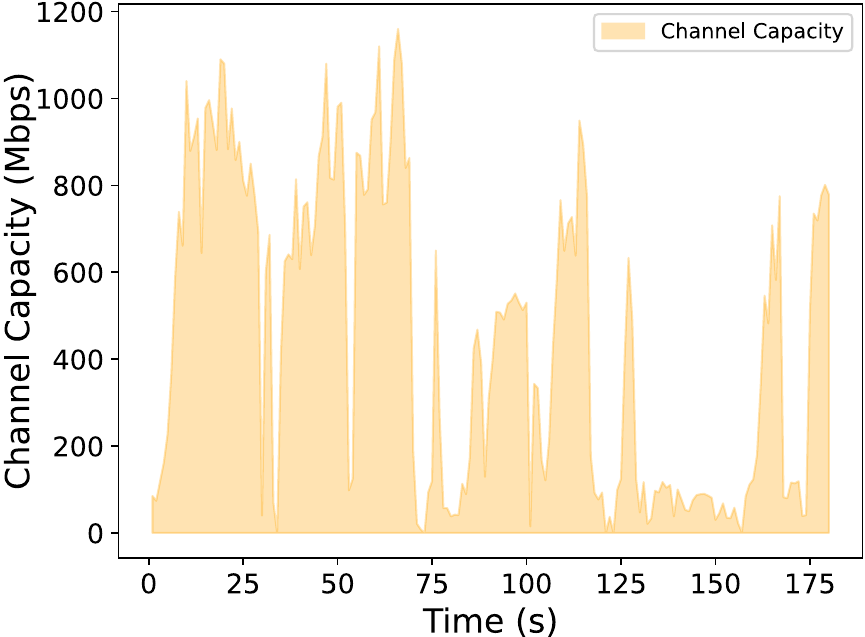}
\end{subfigure}


\begin{subfigure}{.33\textwidth}
  \centering
  \includegraphics[width=0.9\linewidth,height=1.5in]{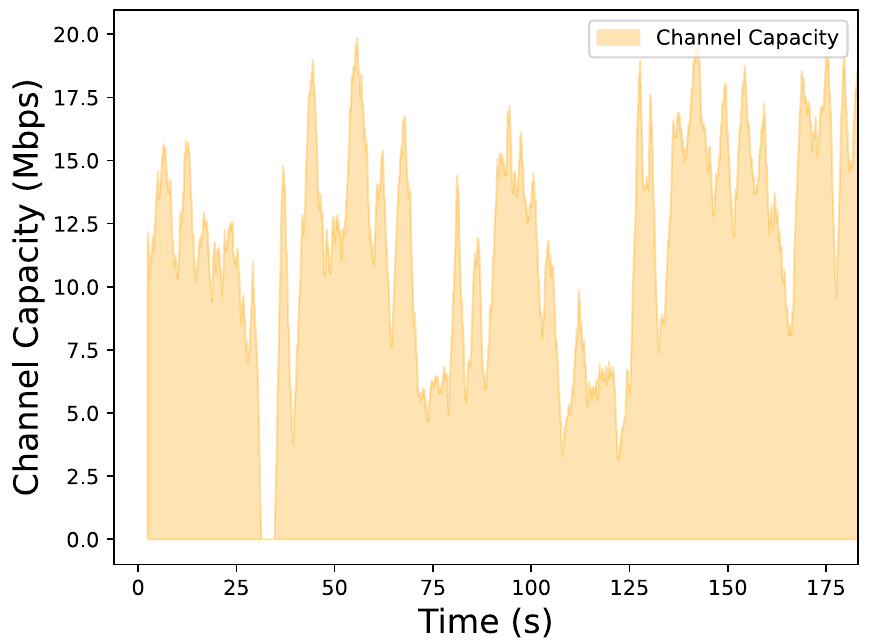}
\end{subfigure}%
\begin{subfigure}{.33\textwidth}
  \centering
  \includegraphics[width=0.9\linewidth,height=1.5in]{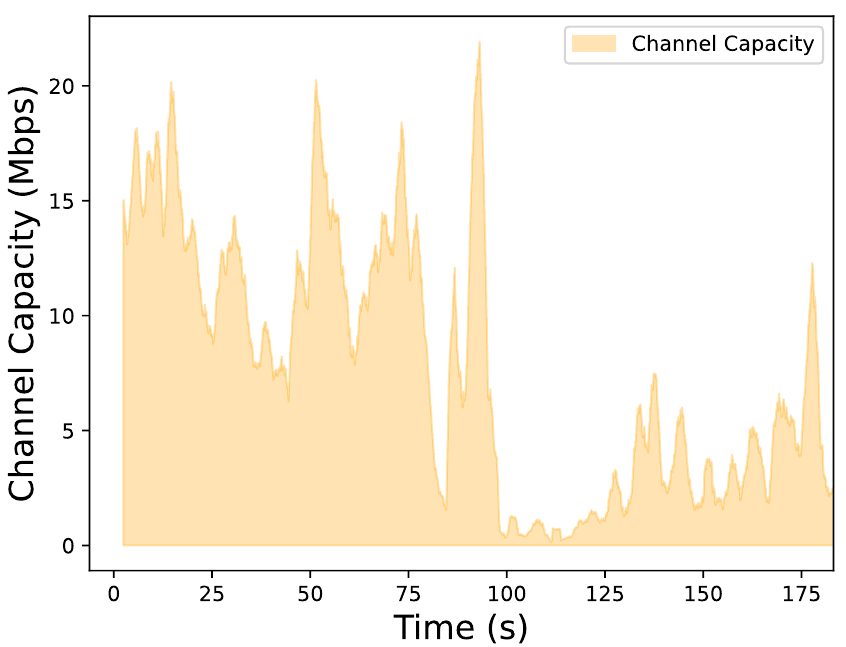}
\end{subfigure}%
\begin{subfigure}{.33\textwidth}
  \centering
  \includegraphics[width=0.9\linewidth,height=1.5in]{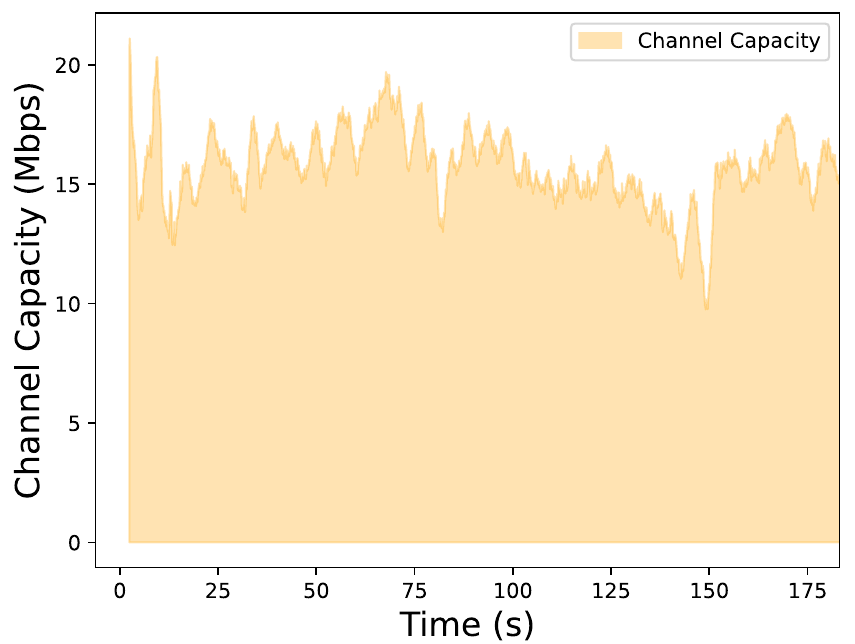}
\end{subfigure}

\begin{subfigure}{.33\textwidth}
  \centering
  \includegraphics[width=0.9\linewidth,height=1.5in]{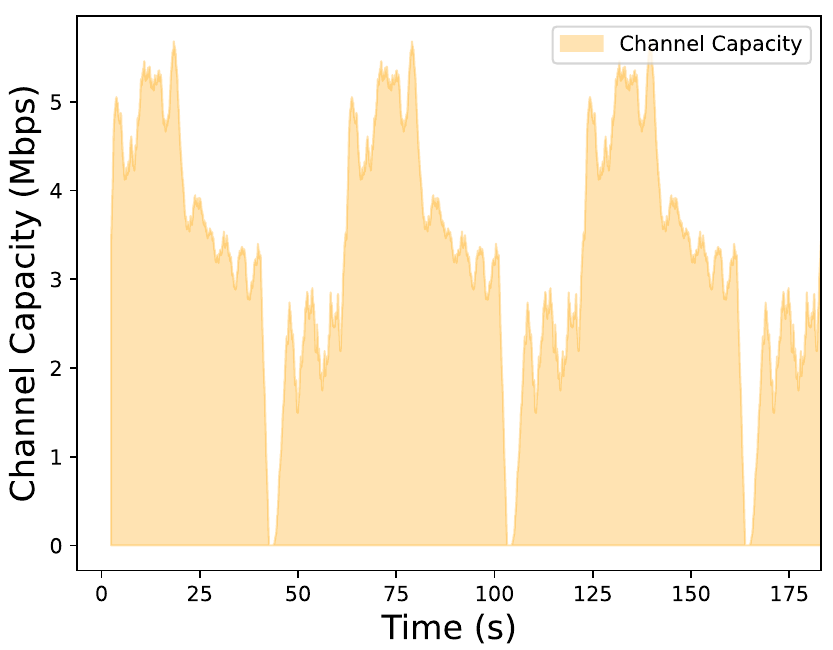}
\end{subfigure}%
\begin{subfigure}{.33\textwidth}
  \centering
  \includegraphics[width=0.9\linewidth,height=1.5in]{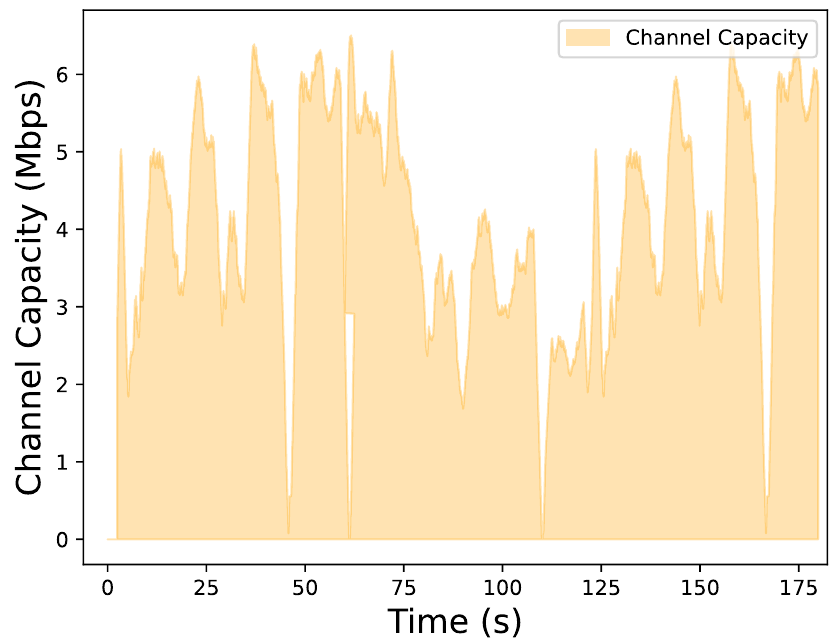}
\end{subfigure}%
\begin{subfigure}{.33\textwidth}
  \centering
  \includegraphics[width=0.9\linewidth,height=1.5in]{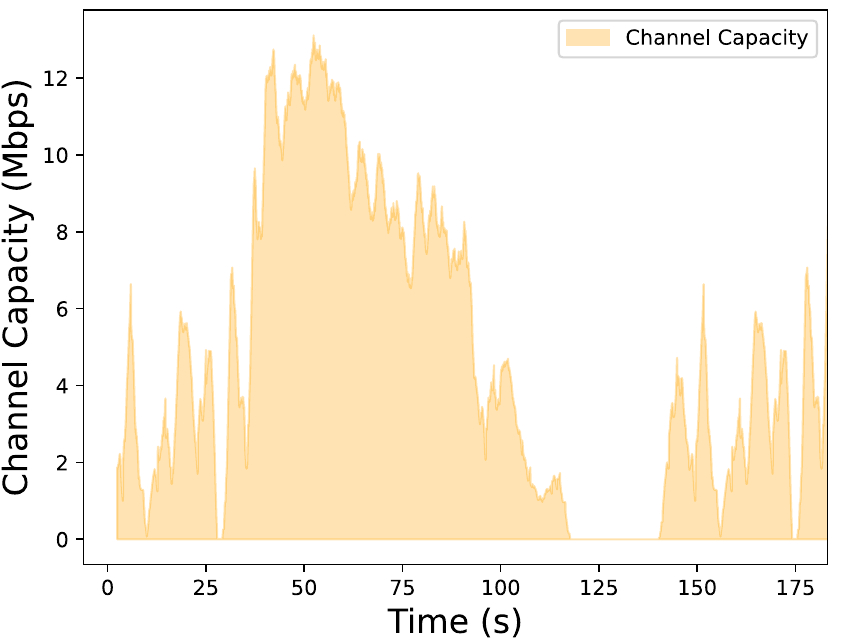}
\end{subfigure}

\caption{Sample of Traces Collected by Prior Works Used in Emulations (5G Traces ( First Row) \cite{narayanan2020lumos5g}, 4G Traces ( Second Row \cite{deepcc}, 3G Traces (Third Row) \cite{c2tcp})}
\label{fig::trace_sample}
\end{figure*}
\subsection{Speed Up Proof}
\label{Speed_up_proof}
\begin{theorem*}
Reminis helps the AIMD logic to reach $w_1$ in $\mathcal{O}(\log{}w_1)$ instead of $\mathcal{O}(w_1)$ in congestion avoidance phase.\\
\end{theorem*}
\begin{proof}
Until reaching $w_1$, in every SI, NCI will infer Zone 1. \\Hence the Non-Deterministic Exploration module will be activated each time during $0\le t\le t_1$. Note that as no queue has been built up yet, the derivative of delay is equal to zero during $0\le t\le T_1$. Therefore, the Gaussian distribution will not change and will continue to work with its default values which are $\mu=1$ and $\sigma^2 = \frac{1}{4}$.
\\ Equation \ref{Gaussain output} shows the expected value of the output of the Non-Deterministic Exploration module.
\begin{equation}
\label{Gaussain output}
\EX[2^{s(x \sim \mathcal{N}(1,\frac{1}{4}))}] \approx 1.5
\end{equation}
Equation \ref{recursive GP} shows that in every $SI_n$, during $0\le t\le t_1$, on average, how the congestion window changes. The additive term ($+1$) is the effect of the underlying AIMD module and the multiplicative term is the effect of the Non-Deterministic Exploration(the Guardian) module.
\begin{equation}
\label{recursive GP}
    cwnd_n = 1.5(cwnd_{n-1} +1 )
\end{equation}
The recursive function mentioned in \ref{recursive GP}, could be written in general form as follow:
\begin{equation}
\begin{split}
    cwnd_n &=  \sum_{i=1}^{n-1} 1.5^{i}\\
    &= \frac{1}{3}(10^{1-n}\times13^{n}-13)
\end{split}
\end{equation}
where:
     $cwnd_1=0$\\

Now let's assume that the SI in which the congestion window will reach $w_1$ is $n_{w_1}$. For finding $n_{w_1}$, its enough to solve:
\begin{equation}
    \frac{1}{3}(10^{1-n_{w_1}}\times13^{n_{w_1}}-13) = w_1
\end{equation} 
It can be shown for large values of $w_1$, solving the above equation we will have:
\begin{equation}
\label{final_ramp-up}
    n_{w_1} < 4log(3w_1)
\end{equation}
As each SI takes $mRTT$, Reminis will on average reach $w_1$ in $4log(3w_1)\times mRTT$, however for an AIMD module, it will take  $w_1\times mRTT$.
\end{proof}

\section{Samples of the 5G Traces Used}
\label{appendix:traces}
Fig. \ref{fig::trace_sample} depicts a sample of traces collected by prior works \cite{narayanan2020lumos5g},\cite{c2tcp} and \cite{deepcc} which are the base of our emulations.
The first row is for 5G traces, the samples in the second row are 4G traces and the final row is for 3G traces. \par As Fig. \ref{fig::trace_sample} shows, all these traces, have a degree of volatility in the access link capacity which makes Reminis a suitable CC algorithm for these networks. However, 5G traces have higher maximum throughput, 87.5 $\times$ and 350 $\times$ compared to 4G and 3G traces respectively. Moreover, 5G traces have more severe changes in the access link capacity which are caused by the changing dynamics of the environment. Two important events are revealed in 5G traces in Fig. \ref{fig::trace_sample}. First is the 4G to 5G or vice versa handovers. In the 5G to 4G handover, the link capacity decreases suddenly from around 1 Gbps to around 100-200 Mbps (the link capacity of 4G links). The opposite happens during 4G to 5G handovers. The second important events are 5G dead zones, in which link capacity drops close to zero. These zones can be seen when a 5G to 5G handover occurs or simply because the user gets out of coverage of any 4G/5G base station. A successful 5G scheme, should be able to react properly to these changes without making any strict assumptions as these events are non-deterministic.

\end{appendices}

\end{document}